\documentclass[
aps,superscriptaddress,
twocolumn,prd,
showpacs,preprintnumbers,nofootinbib,longbibliography,
amsmath,amssymb,
aps,
superscriptaddress]{revtex4-1}
\usepackage[normalem]{ulem}
\usepackage{graphicx}
\usepackage{soul}
\usepackage{amssymb}
\makeatletter
\let\saved@includegraphics\includegraphics
\AtBeginDocument{\let\includegraphics\saved@includegraphics}
\usepackage{url}
\usepackage{color}
\usepackage[dvipsnames]{xcolor}
\usepackage{amsmath}
\usepackage{adjustbox}
\makeatletter

\usepackage{tabularx}
\usepackage{hyperref}
\usepackage[thinc]{esdiff}
\usepackage{physics}
\graphicspath{{./main_Figures/}}
\usepackage{mathtools}
\usepackage{amsthm}

\newcommand{\ncor}[1]{{\textcolor{black}{#1}}}

\newcommand{\etk}{\textsc{Einstein Toolkit }}
\newcommand{\carpet}{\textsc{Carpet }}

\newtheorem{theorem}{Theorem}[section]
\newtheorem{lemma}[theorem]{Lemma}
\newcommand{\nrsur}{\texttt{NRSur7dq4}}

\newcommand{\muB}{\mu_{\rm B}}

\begin{document}

%\title{Did the gravitational-wave signal GW190521g originate in head-on collision of black holes?}

%\title{Some head blowing title :D}
%\title{Constraining the mass of ultralight vector bosons with GW190521\\
%Estimating the mass of a hypothetical ultralight vector boson compatible with GW190521 \\
%\title{GW190521 as a collision of vfector boson stars with mass $\mu_V=6.50\times 10^{-13}$ eV\\
\title{%From the oven to the table skipping the fridge:\\ 
Gravitational-wave parameter inference with the Newman-Penrose scalar $\psi_4$}

%% Notice placement of commas and superscripts and use of &
%% in the author list

\author{Juan Calder\'on~Bustillo}
\email{juan.calderon.bustillo@gmail.com}
\affiliation{Instituto Galego de F\'{i}sica de Altas Enerx\'{i}as, Universidade de
Santiago de Compostela, 15782 Santiago de Compostela, Galicia, Spain}
\affiliation{Department of Physics, The Chinese University of Hong Kong, Shatin, N.T., Hong Kong}
\author{Isaac C.F. Wong}
\email{cfwong@link.cuhk.edu.hk }
\affiliation{Department of Physics, The Chinese University of Hong Kong, Shatin, N.T., Hong Kong}
\author{Nicolas Sanchis-Gual}
\affiliation{Departamento de Astronom\'{i}a y Astrof\'{i}sica, Universitat de Val\`{e}ncia,
Dr. Moliner 50, 46100, Burjassot (Val\`{e}ncia), Spain}
\affiliation{Departamento  de  Matem\'{a}tica  da  Universidade  de  Aveiro  and  Centre  for  Research  and  Development in  Mathematics  and  Applications  (CIDMA),  Campus  de  Santiago,  3810-183  Aveiro,  Portugal}
\author{Samson H. W. Leong}
\affiliation{Department of Physics, The Chinese University of Hong Kong, Shatin, N.T., Hong Kong}
\author{Alejandro Torres-Forn\'e}
\affiliation{Departamento de Astronom\'{i}a y Astrof\'{i}sica, Universitat de Val\`{e}ncia,
Dr. Moliner 50, 46100, Burjassot (Val\`{e}ncia), Spain}
\affiliation{Observatori Astron\`{o}mic, Universitat de Val\`{e}ncia,
C/ Catedr\'{a}tico Jos\'{e} Beltr\'{a}n 2, 46980, Paterna (Val\`{e}ncia), Spain}
\author{Koustav Chandra}
\affiliation{Department of Physics, Indian Institute of Technology Bombay, Powai, Mumbai, Maharashtra 400076, India}
\author{Jos\'e A. Font}
\affiliation{Departamento de Astronom\'{i}a y Astrof\'{i}sica, Universitat de Val\`{e}ncia,
Dr. Moliner 50, 46100, Burjassot (Val\`{e}ncia), Spain}
\affiliation{Observatori Astron\`{o}mic, Universitat de Val\`{e}ncia,
C/ Catedr\'{a}tico Jos\'{e} Beltr\'{a}n 2, 46980, Paterna (Val\`{e}ncia), Spain}
\author{Carlos Herdeiro}
\affiliation{Departamento  de  Matem\'{a}tica  da  Universidade  de  Aveiro  and  Centre  for  Research  and  Development in  Mathematics  and  Applications  (CIDMA),  Campus  de  Santiago,  3810-183  Aveiro,  Portugal}
\author{Eugen Radu}
\affiliation{Departamento  de  Matem\'{a}tica  da  Universidade  de  Aveiro  and  Centre  for  Research  and  Development in  Mathematics  and  Applications  (CIDMA),  Campus  de  Santiago,  3810-183  Aveiro,  Portugal}
\author{Tjonnie G.F. Li}
\affiliation{Department of Physics, The Chinese University of Hong Kong, Shatin, N.T., Hong Kong}
\affiliation{Institute for Theoretical Physics, KU Leuven, Celestijnenlaan 200D, B-3001 Leuven, Belgium}
\affiliation{Department of Electrical Engineering (ESAT), KU Leuven, Kasteelpark Arenberg 10, B-3001 Leuven, Belgium}

\begin{abstract}
    Detection and parameter inference of gravitational-wave signals \ncor{from compact mergers} rely on the comparison of the incoming detector strain data $d(t)$ to waveform templates for the gravitational-wave strain $h(t)$ that ultimately rely on the resolution of Einstein's equations via numerical relativity simulations. These, however, commonly output a quantity known as the Newman-Penrose scalar $\psi_4(t)$ which, under the Bondi gauge, is related to the gravitational-wave strain by $\psi_4(t)=\mathrm{d}^2h(t) / \mathrm{d}t^2$. Therefore, obtaining strain templates involves an integration process that introduces artefacts that need to be treated in a rather manual way. By taking second-order finite differences on the detector data and inferring the corresponding background noise distribution, we develop a framework to perform gravitational-wave data analysis directly using $\psi_4(t)$ templates. We first demonstrate this formalism, and the impact of integration artefacts in strain templates, through the recovery of numerically simulated signals from head-on collisions of Proca stars injected in Advanced LIGO noise. Next, we re-analyse the event GW190521 under the hypothesis of a Proca-star merger, obtaining results equivalent to those in Ref.~\cite{Bustillo:2021proca1}, where we used the classical strain framework. We find, however, that integration errors would strongly impact our analysis if GW190521 was four times louder. Finally, we show that our framework fixes significant biases in the interpretation of the high-mass GW trigger S200114f arising from the usage of strain templates. We remove the need to obtain  strain waveforms from numerical relativity simulations, avoiding the associated systematic errors. 
\end{abstract}

\maketitle

\section{Introduction} The observation of the gravitational-wave (GW) event GW150914 in 2015 by the Advanced LIGO detectors \cite{AdvancedLIGOREF} opened a new window to explore the Universe \cite{GW150914}. In barely half a decade, and after the addition of the Advanced Virgo \cite{TheVirgo:2014hva} and KAGRA \cite{akutsu2020overview} detectors, the number of detections has grown to 90 events, all consistent with the merger of compact objects such as black holes (BHs) and neutron stars (NSs) \cite{GWTC1,GWTC2,GWTC3,GW170817}. These events have provided us with invaluable knowledge about the BH population of our Universe \cite{GWTC2-pop,GWTC3-pop} and their environments, star formation, tests of strong gravity~\cite{TGR_GWTC2,GWTC3-TGR,KerrvsOccam,Carullo:2019flw,Isi2019_nohair,Giesler2019_overtones,Ghosh2021,Ma_method,Ma_PRD} and the observation of new strong-field phenomena \cite{CalderonBustillo:2018zuq,CalderonBustillo:2019wwe,Varma_strong_kick,Measured_kick_magnitude_GW190814,GW190412_kick} to name a few. The retrieval of this information relies on an accurate extraction of the parameters of the GW source. This is commonly carried out through the comparison of the incoming strain detector data $d(t)$ to pre-computed waveform templates for the gravitational-wave strain $h(t; \theta)$ \cite{MatchedFilter,Allen:2005fk,Lalinference,Ashton:2018jfp,GW150914_properties} that span a continuous range of possible source parameters $\theta$ such as the masses and spins of merging compact objects. In this process, it is crucial that waveform templates are faithful representations of the incoming GWs. For the case of the GWs emitted during the early inspiral, these templates can be obtained through analytical approximated techniques such as post-Newtonian approximations \cite{Blanchet2014} or effective-one-body formalisms \cite{Damour:2000zb,Buonanno:1998gg}. However, modelling the full space-time dynamics taking place during the merger and ringdown stages of compact binary mergers requires solving the full Einstein's equations for the system, which can only be done through numerical simulations using numerical relativity (NR) \cite{Pretorius:2005gq,Bruegmann:2006at,Campanelli:2005dd,Baker:2006,Boyle:2007ft,Jani:2016wkt,RITCatalog,SXSCatalog,NRAR,Lehner2014_Review,CoRE,Gonzalez2023}. Consequently, approximated models are commonly calibrated to these simulations during the merger and ringdown stages \cite{Khan:2015jqa,Husa:2015iqa,XPHM_Pratten,XPHM_Cecilio,PhenomPv3HM,SEOBRNv4PHM,Gamba2022_TEOBResumS}. Alternatively, ``surrogate'' models can also be constructed by interpolating through a given set of NR simulations \cite{Blackman:2017pcm,NRSur7dq4,Freitas:2022xvg}. 

\begin{figure}[ht]
    \includegraphics[width=0.48\textwidth]{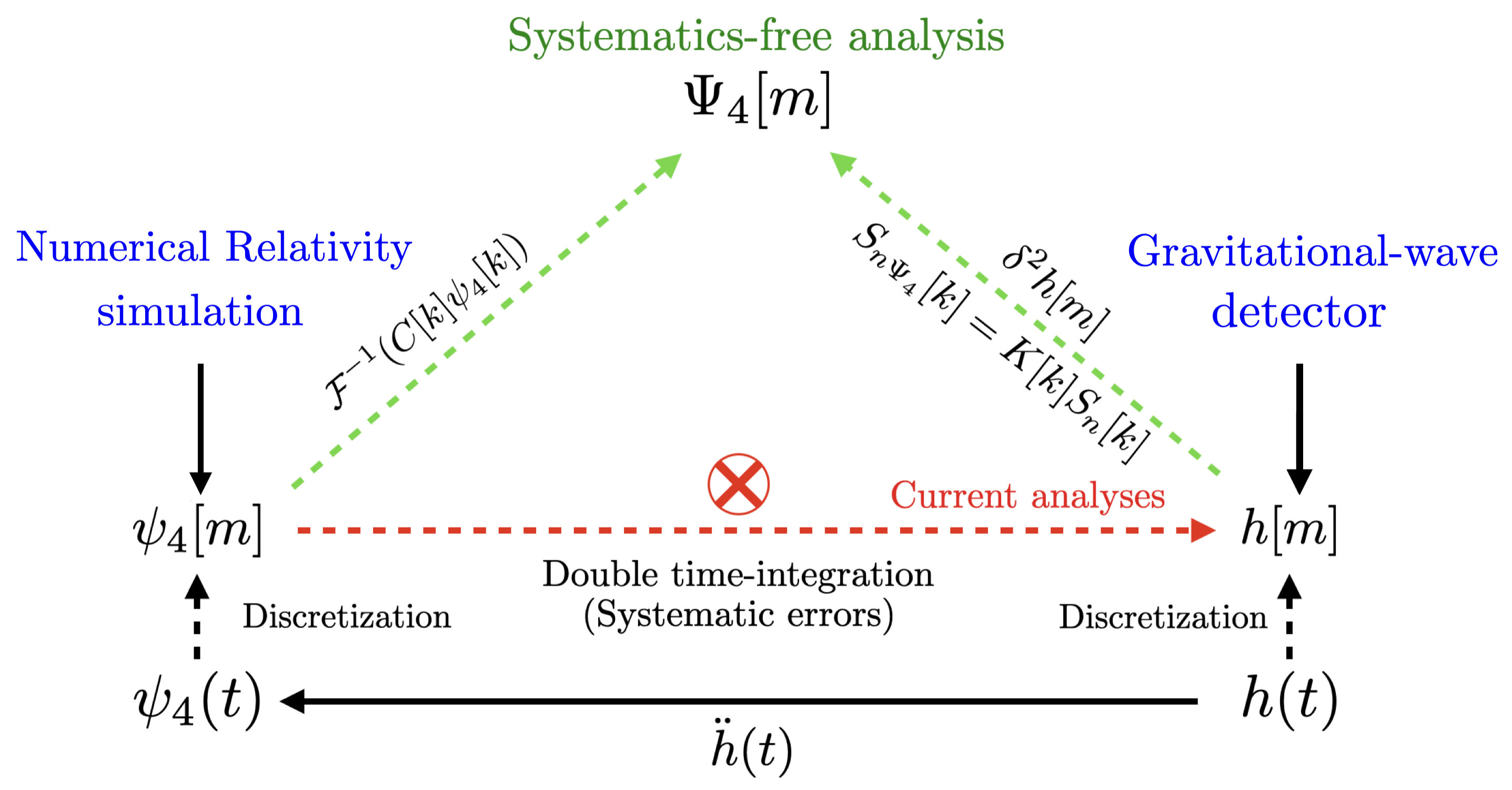}
    \caption{
    \textbf{Schematic comparison of our proposed data analysis framework and the currently used one.} To date, the $\psi_4$ magnitude outputted by numerical relativity simulations is converted to the strain $h$ outputted by gravitational-wave detectors via a double integration that is subject to systematic errors (red path). Instead, we transform both the simulation $\psi_4$ and the detector $h$ (and power-spectral-density $S_n$) into a third quantity that we label by $\Psi_4$, avoiding the integration process and the corresponding systematic errors (green paths). 
    }
    \label{fig:diagram}
\end{figure}

When available, NR provides the most accurate prediction for the GW emission of a given source. Therefore, if simulations are available in the parameter space of interest, a direct comparison of GW data to NR simulations is fundamental to -- at least -- check the robustness of the results provided by approximated models \cite{GW150914_NR,GW190521I}. Furthermore, in some cases such as highly eccentric or precessing sources, continuous semi-analytical models may not exist, leaving NR as the only option to analyse the data. Consequently, several studies have directly compared some of the existing signals to NR templates \cite{GW150914_NR,GW190521I,Lange:2017wki,Gayathri2022_ecc_NatAstro} and even used the latter as simulated signals to evaluate the efficacy of parameter estimation and detection algorithms~\cite{Aylott2009,ninja2,Chandra2020,Bustillo2021_ecc,Williamson2017}. 

Continuous models that include the entire inspiral-merger-ringdown process can only be built for regions of the parameter space densely covered by the available numerical simulations, namely sources with small orbital eccentricity and relatively equal masses. While these examples have sufficed to explain all GW signals detected to date, we are entering an era in which comparison with more exotic scenarios, for which only NR waveforms exist, is in order. Moreover, as of now, NR provides the only way to accurately model the dynamics of exotic compact objects and search for new physics beyond the neutron-star and Kerr black hole paradigm. 
%and e.g., the post-merger emission from neutron-star mergers. 
For example, in Ref.~\cite{Bustillo:2021proca1} we recently compared GW190521 to numerical simulations of head-on mergers of exotic horizon-less objects known as Proca stars, demonstrating that the latter scenario is slightly more consistent with the data than the standard one based on BBH mergers.

\section{Extraction of gravitational-wave strain from numerical simulations via the Newman-Penrose scalar $\psi_4$}

While current GW detectors output a quantity known as GW strain $h(t)$, \ncor{the most extended type of NR simulations, which are based on the 3+1 formulation (or Cauchy evolution), return the so-called Newman-Penrose (NP) scalar $\psi_4(t)$} \cite{Newman:1961qr,Campanelli1999}. Under the Bondi gauge, the scalar $\psi_4(t)$ is related to the strain by $\psi_4(t)=\dv*[2]{h(t)}{t}$ \cite{Bondi}. Obtaining $h(t)$ therefore requires a double time integration, which is a non-trivial process involving fundamental difficulties. A well-known effect is the appearance of non-linear drifts in the resulting strain-waveform arising from the \textit{time-domain integration} of finite length, discretely sampled and noisy data streams. These are independent of the parameters of the simulation, such as gauge or numerical method used \cite{Pollney_Reissweig}.\\ 

\textit{Frequency-domain integration} methods can avoid the effects arising from time-domain integration but at the cost of modifying the original data. One of the best-known effects is the impact of spurious low-frequency modes in the strain waveform. It is known that the effect of these modes, resulting from either spectral leakage or aliasing effects, can be significantly suppressed through the usage of high-pass signal filters \cite{Pollney_Reissweig,HighPass_Santamara2010} that can reduce the energy within frequencies lower than a chosen cutoff $\omega_0$. This technique is commonly known as ``fixed-frequency integration'' (FFI).
A common and well-motivated choice for $\omega_0$ is that corresponding to the lowest instantaneous frequency of the GW emission. This is, for instance, the strategy used by the existing parameter-estimation code RIFT \cite{Lange:2017wki} to directly compare GW data to NR templates.
In practice, however, the choice of $\omega_0$ requires a certain amount of tuning. On the one hand, a small value will amplify nonphysical low-frequency components during the integration process. On the other hand, a large value may suppress the physical frequencies of the waveform. In some cases, such a choice can be clearly guided by the known features of the ``true'' signal.
For instance, in the case of quasi-circular binaries, the inspiral GW `chirp' frequency is a monotonically increasing function of time, which provides a natural way to associate a given choice of $\omega_0$ with a given starting time for the strain waveform.

This choice, however, is not obvious or even well defined for cases where the GW frequency is not a monotonic function of time. On the one hand, this makes FFI itself a potential source of error that, as we will show, can qualitatively impact the interpretation of the source. On the other hand, in the best-case scenario, a ``correct'' obtention of the strain data requires an artisan and time-consuming trial-and-error process adapted to each particular type of source. This is the case of some of the most interesting sources that the astrophysical community is trying to detect for the first time in the next observing run of LIGO and Virgo. Such cases include eccentric mergers \cite{Gayathri2022_ecc_NatAstro}, highly precessing mergers, dynamical captures, \cite{GW190521_Gamba} or even cases in which an inspiral stage does not exist at all. This is also the situation for the (academic) case of head-on collisions \cite{sanchis2019head} that we will address in this work or for core-collapse supernovae waveforms, for which the bounce GW signal consists essentially of a burst~\cite{Dimmelmeier:2002,Cerda:2013,Richers:2017}. Finally, the act of integrating involves a choice of integration constants that can cause fundamental changes in the properties of the waveform. For instance, it is typically imposed that the average of the GW strain should be zero, which automatically removes/changes the effect of 
GW memory \cite{Lasky2016}.\\

\ncor{\textit{Integration-free extraction methods.} While in this work we focus on numerical waveforms obtained through the NP formalism, whose use is much generalised in the numerical relativity community, we note that there exist alternative methods that can directly extract the GW strain. First, within the Cauchy evolution framework, the GW strain can be also directly extracted at finite radii through the so-called Regge-Wheeler-Zerilli (RWZ) method \cite{Regge1957,Zerilli1970,Moncrief1974}. Just as in the NP case, waveforms can then be extrapolated to null infinity through e.g.~polynomial expansions \cite{Boyle:2009vi,SXSCatalog}. RWZ extraction has been long employed by the SXS collaboration \cite{SXSCatalog}, producing extensive catalogs of BBH waveforms that include the largest number of inspiral cycles in the literature. These waveforms have been consistently used to inform  continuous waveform models (e.g.~\cite{NRSur7dq4}), often in combination with NP ones (e.g.~\cite{SEOBRNv4PHM,PhenomPv3HM}), and to directly analyse several GW events \cite{GW150914_NR,GW190521I}. Second, within the so-called Cauchy Characteristic Extraction (CCE) framework~\cite{Bishop1996}, waveforms can be directly extracted at future null infinity. While early simulations were limited to output the Bondi news 
function~\cite{Winicour} (given by ${\cal{N}}(t)=\dv*{h(t)}{t}$), later developments made possible the direct extraction of the GW strain \cite{Bishop2013_strain,Pollney2011_Llama,2110.08635_SXS_CCE}, including some simulations in the SXS catalog \cite{2110.08635_SXS_CCE}}.\\

\ncor{While the above approaches present advantages with respect to the NP formalism, there are reasons that motivate the generalised used of the latter. First, the master wave equations of the RWZ approach are obtained under the assumption that the background metric can be described by a Schwarzschild spacetime where perturbations are applied. Also, it has been found that the relative accuracy of RWZ and NP methods depend on the the case of application, e.g., \cite{Reisswig2011_RWZ_test,DePietri2016}. Second, while CCE can deliver extremely accurate waveforms, it involves specific complications that have so far prevented its widespread use, on top of its intensive computational cost. For instance, recent studies have discussed the weak hyperbolicity of the characteristic evolution equations~\cite{giannakopoulos2020hyperbolicity,giannakopoulos2022gauge,giannakopoulos2023numerical}. Also, initial spurious emission known as junk radiation has been found to last significantly longer in these 
simulations~\cite{2110.08635_SXS_CCE}. 
In addition, characteristic evolution (based on null foliations) is known to be limited to describe BBH systems, as null surfaces may focus and form caustics~\cite{Winicour}.
This triggered the design of so-called Cauchy Characteristic Matching methods \cite{CCM_Bishop,giannakopoulos2023numerical}. Finally, we note that NP presents several advantages of its own, as enumerated in \cite{Campanelli1999}: a) it provides a first-order, gauge-invariant description of the radiation field \textcolor{black}{(please see Koop \& Finn \cite{KoopFinn} for a fully gauge invariant derivation of the detector response)}; b) it does not rely on any frequency or multipole decomposition; c) the Weyl scalars ($\psi_4$ among them) are defined in the full non-linear theory. A one-parameter perturbative expansion of this theory was proved
to provide a reliable account of the problem \cite{Damour1994}; and d), finally, the NP formulation provides a simpler framework to organise higher-order perturbation schemes. For an overview of different waveform extraction methods we refer the reader to \cite{Bishop2016}}\\

In this work we remove from GW data analysis the fundamental problems related to waveform integration by avoiding such step. We present a framework, schematically summarised in Fig.~\ref{fig:diagram}, to perform GW data analysis directly using $\psi_4(t)$. This provides an uniquely defined way to obtain GW waveforms for data analysis -- free of human choices -- which are by definition free of the systematic errors related to waveform integration. We showcase our framework in the context of parameter inference and model selection performed on both synthetic signals injected in LIGO-Virgo noise and on real GW signals.\\

\section{Our case of study:\\ Mergers of Proca-stars}

\subsection{Proca stars and Dark Matter}

Proca stars belong to a family of theoretical exotic compact objects (ECOs) known as bosonic stars~\cite{Schunck:2003kk,brito:2016proca,herdeiro:2017vs,herdeiro:2019flat}. 
These are part of the wider family of objects known as ``BH mimickers'' which, lacking the characteristic event horizon of BHs, can reproduce many of their properties - see $e.g.$~\cite{Cardoso:2019rvt,Vincent:2015xta,Herdeiro:2021lwl}, avoiding issues related to the black hole singularity, as well as poorly understood issues related to quantum fields near event horizons~\cite{Almheiri:2012rt}.
ECOs have been proposed, $e.g.$, as dark matter candidates \cite{EuCAPT_astroparticle}, in particular in models invoking the existence of hypothetical ultralight ($i.e.$~sub-eV) bosonic particles, often referred to as fuzzy dark matter~\cite{Hui:2016ltb}. One common candidate is the pseudo-scalar QCD axion~\cite{Peccei:1977hh}, but other ultralight bosons arise, e.g., in the string axiverse \cite{Arvanitaki2010}. In particular, vector bosons are also motivated in extensions of the Standard Model of elementary particle \cite{Freitas:2021cfi} and can clump together forming macroscopic entities dubbed vector boson stars or Proca stars. 

Bosonic stars are amongst the simplest and dynamically more robust ECOs proposed so far and their dynamics have been extensively studied,  $e.g.$ \cite{liebling2017dynamical,bezares2017final,palenzuela2017gravitational,sanchis2019head}. Scalar boson stars and their vector analogues, Proca stars~\cite{brito:2016proca,sanchis2017numerical}, are self-gravitating stationary solutions of the Einstein-(complex, massive) Klein-Gordon~\cite{Schunck:2003kk} and of the Einstein-(complex) Proca~\cite{brito:2016proca} systems, respectively. These consist of complex bosonic fields oscillating at a well-defined frequency $\omega$, which determines
%gives 
the mass and compactness of the star. %Unlike other ECOs, 
Bosonic stars can dynamically form without any fine-tuned condition through the gravitational cooling mechanism \cite{Seidel1994,DiGiovanni2018}.

While spinning solutions have been obtained for both scalar and vector bosons, the former are unstable against non-axisymmetric perturbations, in the simplest models wherein the bosonic field is free~\cite{SanchisGual2019,di2020dynamical}. Hence, we will focus on the vector case in this work. 
For non-self-interacting bosonic fields, the maximum possible mass of the corresponding stars is determined  by the boson particle mass $\mu_{\rm B}$.  In particular, ultralight bosons within $10^{-13}\leq\mu_{\rm B}\leq10^{-10}$~eV, can form stars with maximal masses ranging between ${\sim}1000$ and 1 solar masses, respectively. In \cite{Bustillo:2021proca1}, we showed that GW190521 was consistent with the head-on collision of two Proca stars with $\mu_{\rm B}=8.7\times 10^{-13}$ eV.\\

\subsection{Numerical simulations of Proca star mergers}

We will demonstrate our $\psi_4$ data analysis making use of NR simulations of head-on collisions of spinning Proca stars.
In addition to the quadrupole $(\ell,m)=(2,\pm 2)$ modes dominant for circular mergers, our simulations also yield the $(2,0)$ mode, co-dominant for the case of head-on collisions, and the much weaker $(3,\pm3)$ and $(3,\pm2)$ modes.
%% Please check the citations in these two paragraphs are the intended ones.
Our set of waveforms is obtained from simulations of the collisions of two spinning Proca stars with aligned spin axes~\cite{Bustillo:2021proca1, Sanchis:synchronize,SanchisGual2022_Phases}. %This sort of collisions were recently studied in Ref.~\cite{Bustillo:2021proca1,  Sanchis:synchronize}. 
Although starting from rest, the trajectories of the two stars are eccentric rather than strictly head-on due to frame-dragging.
%Although the stars start from rest, due to frame dragging the binary describes an eccentric (rather than precisely head-on)trajectory. 
%The end point depends on the progenitor Proca stars. In the region of the parameter space explored here, the Proca star progenitors are sufficiently massive to trigger black hole formation after the merger.
In our study's region of parameter space, all Proca-star progenitors are sufficiently massive and compact to trigger the gravitational collapse of the remnant. Therefore, the outcome of the collision always leads to BH formation after the merger.
The simulations are performed with the \etk infrastructure~\cite{Loffler:2011ay, Loffler:2012et, Miguel:2013intro}, together with the \carpet package~\cite{Schnetter:2004mesh, Goodale:2003cactus} for mesh-refinement. The Proca evolution equations are solved via a modified Proca thorn~\cite{zilhao2015nonlinear,SanchisGual2019, sanchis2019head, Helvi:2020canuda} to include a complex field. We consider both equal-mass and unequal-mass cases, as reported in our numerical Proca catalogue \cite{SanchisGual2022_Phases}. The initial data consists of
the superposition of two equilibrium solutions separated by $D = 40/\mu$~\cite{palenzuela2007head,palenzuela2017gravitational,sanchis2019head,Bustillo:2021proca1}, in geometrized units, guaranteeing an admissible initial constraint violation. The equilibrium stars are numerically constructed using the solver {\tt fidisol/cadsol} for non-linear partial differential equations of elliptic type, via a Newton-Raphson method (see~\cite{brito:2016proca, herdeiro:2017vs, herdeiro:2019flat} for more details). %We are currently extending our Proca catalogue using a generative adversarial network to produce synthetic waveforms without the need of intensive numerical relativity simulations~\cite{Freitas:2022xvg}.

\section{Data analysis with $\psi_4$}
\label{sec:psi4}
Consider an observation model
\begin{equation}
    \label{eq:obs_model}
    \begin{split}
        d(t)
    &= F_{+}h_{+}(t) + F_{\times} h_{\times} (t) + n(t) \\
    &= s(t) + n(t)\,,
    \end{split}
\end{equation}
where $s(t) = F_{+}h_{+}(t) + F_{\times} h_{\times} (t)$ is the GW strain, $F_{+}$ and $F_{\times}$ are the beam pattern functions of the $+$ and $\times$ polarization states i.e.\ $h_{+}(t)$ and $h_{\times}(t)$ respectively, and $n(t)$ is the detector noise. Here we only consider a transient signal, therefore the beam pattern functions are approximated to be constant over the duration of the signal for a given sky localisation and polarisation angle. We can rewrite Eq.~\eqref{eq:obs_model} as follows
\begin{equation}
    \begin{split}
        d(t) &= \Re\left[(F_{+} + iF_{\times})(h_{+}(t) - ih_{\times}(t))\right] + n(t) \\
             &= \Re\left[(F_{+} + iF_{\times})h(t)\right] + n(t)\,,
    \end{split}
\end{equation}
where $h(t) = h_{+}(t) - ih_{\times}(t)$. Taking second-time derivative on both sides yields
\begin{equation}
    \begin{split}
        \diff[2]{d(t)}{t}
    &=
    \Re\left[(F_{+} + iF_{\times})
    \psi_{4}(t)\right] + 
    \diff[2]{n(t)}{t} \\
    &= s_{\psi_{4}}(t) + \diff[2]{n(t)}{t} \,, \\
    \end{split}
\end{equation}
where $\psi_{4}(t) = \dv*[2]{h(t)}{t}$ and $s_{\psi_{4}}(t) = \dv*[2]{s(t)}{t}$. Now we have obtained the observation model with $\psi_{4}(t)$ directly involved. Since in practice we analyse the digital strain data which are discrete, we have to replace the second-order differential operator $\dv*[2]{t}$ by the second-order difference operator $\delta^{2}$ defined by

\begin{equation}
    (\delta^{2}x)[m] \coloneqq \frac{x[m + 1] - 2x[m] + x[m - 1]}{(\Delta t)^{2}} \,,
    \label{eq:second_order_finite_difference}
\end{equation}

where $x[m]$ is a discrete time series (labelled by index $m$) with a sampling interval $\Delta t$. We then have

\begin{equation}
    (\delta^{2}d)[m] =
    \Re\left[(F_{+} + iF_{\times}) (\delta^{2}h)[m]\right] +
    (\delta^{2}n)[m] \,.
    \label{eq:obs_model_discrete2}
\end{equation}

To express the above observation model with a closer notation connection to $\psi_{4}(t)$, i.e., the second derivative of $h(t)$, we put a subscript $\Psi_{4}$ to represent a second-order differenced time series i.e.\ $x_{\Psi_{4}}[m] := (\delta^{2}x)[m]$. And we also reserve $\Psi_{4}[m]$ as a special notation for $(\delta^{2}h)[m]$, in analogy with $\psi_{4}(t)$. With the new set of notations, we rewrite Eq.~\eqref{eq:obs_model_discrete2} as
\begin{equation}
    \label{eq:obs_model_discrete2_psi4}
    \begin{split}    
        d_{\Psi_{4}}[m] 
    &= \Re \left[
        (F_{+} + iF_{\times})\Psi_{4}[m]
    \right]
    +
    n_{\Psi_{4}}[m] \\
    &= s_{\Psi_{4}}[m] + n_{\Psi_{4}}[m]
    \end{split}
\end{equation}

where $s_{\Psi_{4}}[m] =\Re \left[(F_{+} + iF_{\times})\Psi_{4}[m] \right]$ is the second difference of the GW strain signal. In practice, parameter estimation is often performed on the data in the Fourier domain due to the more efficient evaluation of the likelihood function as compared to that in the time domain (see e.g.\ Ref.~\cite{Carullo:2019flw}). Applying Fourier transform on Eq.~\eqref{eq:obs_model_discrete2_psi4} yields
\begin{equation}
    \widetilde{d}_{\Psi_{4}}[k] =
    \frac
    {
            (F_{+} + iF_{\times}) \widetilde{\Psi}_{4}[k] + 
            (F_{+} - iF_{\times}) \widetilde{\Psi}^{*}_{4}[-k]
    }
    {2}
    +
    \widetilde{n}_{\Psi_{4}}[k].
\end{equation}

We note, however, that since the $\psi_4(t)$ extracted from NR simulations is sampled from the second derivative of $h(t)$, $\ddot{h}[m]$, \textit{but not} the second-order finite difference of the discrete strain $(\delta^2 h)[m]$; we cannot directly use $\psi_4$ as templates for our analysis. Instead, as represented on the left side of Fig.~\ref{fig:diagram}, we need to transform these following the relation
\begin{equation}
    \widetilde{\Psi}_{4}[k] = \frac{1-\cos(2\pi k\Delta f\Delta t)}{2\pi^{2}(k\Delta f\Delta t)^{2}}\,\widetilde{\psi}_{4}(k\Delta f)\,,
    \label{eq:psiPsi}
\end{equation}
for which we provide the proof in Appendix~\ref{app:relation}.  In the above equation, $\Delta f = 1 / (M\Delta t)$ and $M$ is the length of the discrete $\Psi_{4}[m]$. We finally obtain the observation model in the Fourier domain with $\widetilde{\psi}_{4}$ directly involved as follows
\begin{equation}
    \widetilde{d}_{\Psi_{4}}[k] = 
    \widetilde{s}_{\Psi_{4}}[k; \theta] +
    \widetilde{n}_{\Psi_{4}}[k]
\end{equation}
where
\begin{widetext}
\begin{equation}
    \widetilde{s}_{\Psi_{4}}[k; \theta = \{\alpha, \delta, \psi, t_{\text{event}}, \theta'\}] =
    \frac{1-\cos(2\pi k \Delta f \Delta t)}{4\pi^{2}(k\Delta f \Delta t)^{2}}
    \left[
        (F_{+} + iF_{\times}) \widetilde{\psi}_{4}(k\Delta f; \theta') + 
        (F_{+} - iF_{\times}) \widetilde{\psi}^{*}_{4}(-k\Delta f; \theta')
    \right]
\end{equation}
\end{widetext}
where $F_+$ and $F_{\times}$ are functions of the sky location of the source, i.e.\ the right ascension $\alpha$, the declination $\delta$, the polarisation angle $\psi$, and the event time $t_{\text{event}}$, and $\theta'$ are other source parameters.

Another crucial ingredient for parameter estimation is the distribution of the second-differenced noise $\widetilde{n}_{\Psi_{4}}[k]$ in order to obtain the likelihood function. It can be shown (see Appendix~\ref{app:distr_delta_2}) that if $n(t)$ follows the stationary Gaussian distribution with power spectral density $S_{n}(f)$, then $n_{\Psi_{4}}[m]$ also follows the stationary Gaussian distribution with power spectral density $S_{n_{\Psi_{4}}}[k]$ as
\begin{equation}
    \label{eq:psi4_psd}
    S_{n_{\Psi_{4}}}[k] = \frac{1}{(\Delta t)^{4}}
    \left(
        6 -
        8 \cos(\frac{2\pi k}{M}) +
        2 \cos(\frac{4\pi k}{M})
    \right)
    S_{n}[k]
\end{equation}
where $S[k]$ is understood to be $S(k \Delta f)$. 
%with $\Delta f = 1 / (M\Delta t)$.

The likelihood function for source parameters $\theta$ given the second-differenced strain data $d_{\Psi_{4}}$ in the Fourier domain is therefore
\begin{equation}
    \label{eq:logl}
     \mathcal{L}(\ncor{d_{\Psi_{4}} \mid \theta})
    \propto
    \exp \bigg{[} -\frac{1}{2}
    (
    d_{\Psi_{4}} - s_{\Psi_{4}}(\theta)|
    d_{\Psi_{4}} - s_{\Psi_{4}}(\theta)
    ) \bigg{]}\,,
\end{equation}
where $s_{\Psi_{4}}(\theta)$ is the second-differenced template with parameters $\theta$, and $(a | b)$ denotes the noise-weighed inner product defined as \cite{Cutler:1994ys}
\begin{equation}
    \label{eq:inner}
    (a | b) = 4 \Re \int_{f_\text{min}}^{f_{\text{max}}}
    \frac
    {\tilde{a}^{*}(f)\tilde{b}(f)}
    {S_{n_{\Psi_{4}}}(f)}
    \,\dd f\,,
\end{equation}
with $S_{n_{\Psi_{4}}}(f)$ the power spectral density of the second-differenced detector noise given in Eq.~\eqref{eq:psi4_psd}.

\subsection{Summarised recipe for a $\psi_4$-analysis}
We summarise here our method to perform GW data analysis with $\psi_4$. 
Consider the canonical situation where we have detector strain data $d(t)$, the corresponding PSD $S_n(f)$ and strain templates $s(t;\theta)$ for source parameters $\theta$. Then, an analysis based on the Newman-Penrose scalar can be implemented by just replacing:
\begin{equation}
    \begin{array}{rcl}
        \text{Data: } & d(t)  \rightarrow d_{\Psi_4}(t) \equiv (\delta^{2}d)(t) &\text{(Eq.~\eqref{eq:second_order_finite_difference})} \\
        \text{PSD: } & S_n(f) \rightarrow S_{n_{\Psi_4}}(f) 
                     &\text{(Eq.~\eqref{eq:psi4_psd})} \\
        \text{Templates: } & \tilde{s}(f;\theta) \rightarrow {\widetilde{s}_{\Psi_4}(f;\theta)}
                           &\text{(Eq.~\eqref{eq:psiPsi})}.
    \end{array}
    \label{eq:transformation}
\end{equation}
Above, we assume that $\Psi_4$ templates are obtained from the $\psi_4$ outputted by NR simulations through Eq.~\eqref{eq:psiPsi}, i.e., following the right path of Fig.~\ref{fig:diagram}, and are therefore free of integration systematics present in the strain templates. Consequently, both sides of Eq.~\eqref{eq:transformation} generally lead to different results, which is the point of our work. Nevertheless, one can also check that obtaining $\Psi_4$ as $\delta^{2} h$ following the left side of \ncor{Fig.}~\ref{fig:diagram} (taking second-order finite differences on the strain templates) makes both sides of Eq.~\eqref{eq:transformation} return identical results. %\cor{Another possibility could be to write the detector’s response  directly in terms of the Riemann tensor (and thus $\psi_4$), as shown by Koop and Finn~\cite{koop2014physical}.}

%In reality there are two possible ways to obtain the templates $\tilde{s}_{\Psi_4}(f)$, which will lead to different situations. First, one can simply convert existing strain templates, subject to integration errors, $s(t)$ into $\delta^2 s(t)$, in which case the two sides will \label{eq:transformation} yield identical results. Second, in order to take advantage of this formalism, one can directly convert 
%we note that since the obtention of the strain $h(t)$ from the Newman-Penrose scalar $\psi_4(t)$ is subject to integration systematics, the two sides of Eq.~\eqref{eq:transformation} may lead to different results arising from the substitution of $\tilde{s}(f)$ by $\tilde{s}_{\Psi_4}(f)$, which is indeed the point of this work. Identical results, however, are  obtained by the following replacement in Eq.~\eqref{eq:transformation} 
%\begin{equation}
%    \begin{aligned}
%        \psi_4(t) &\rightarrow  (\delta^{2}h)(t).
%    \end{aligned}
%\end{equation}
%where $(\delta^{2}h)(t)$ is the second difference of the strain template $h(t)$, potentially affected by the aforementioned systematics.

\section{Results: Data whitening}
GW data analysis ultimately relies on whitened data. This is, the detector data divided by the amplitude spectral density of the background noise $\tilde{d}(f)/\sqrt{S_n(f)}$. This is then matched-filtered \cite{MatchedFilter} with the whitened waveform templates $\tilde{h}(f)/\sqrt{S_n(f)}$. Here we show that the transformations we perform on both of the strain data and PSD to obtain their $\Psi_4(t)$ versions lead to identical whitened data and templates. Therefore, these lead to a completely equivalent analysis where the only difference is that $\Psi_4(t)$-templates are free of systematic errors introduced during the obtention of the $h(t)$ templates through integration.\\ 

\begin{figure*}
    \begin{center}
        \includegraphics[width=0.54\textwidth]{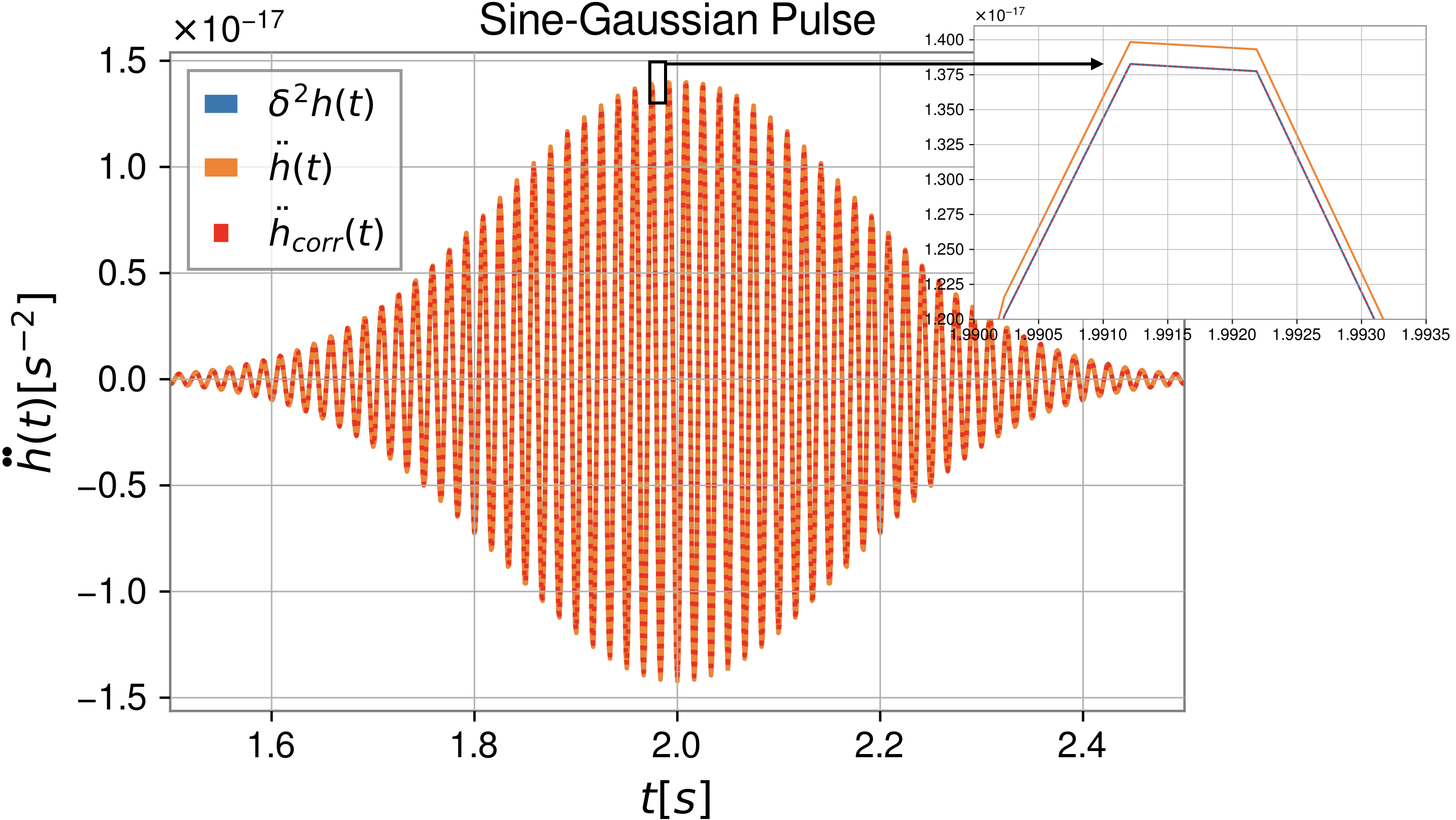}
        \includegraphics[width=0.45\textwidth]{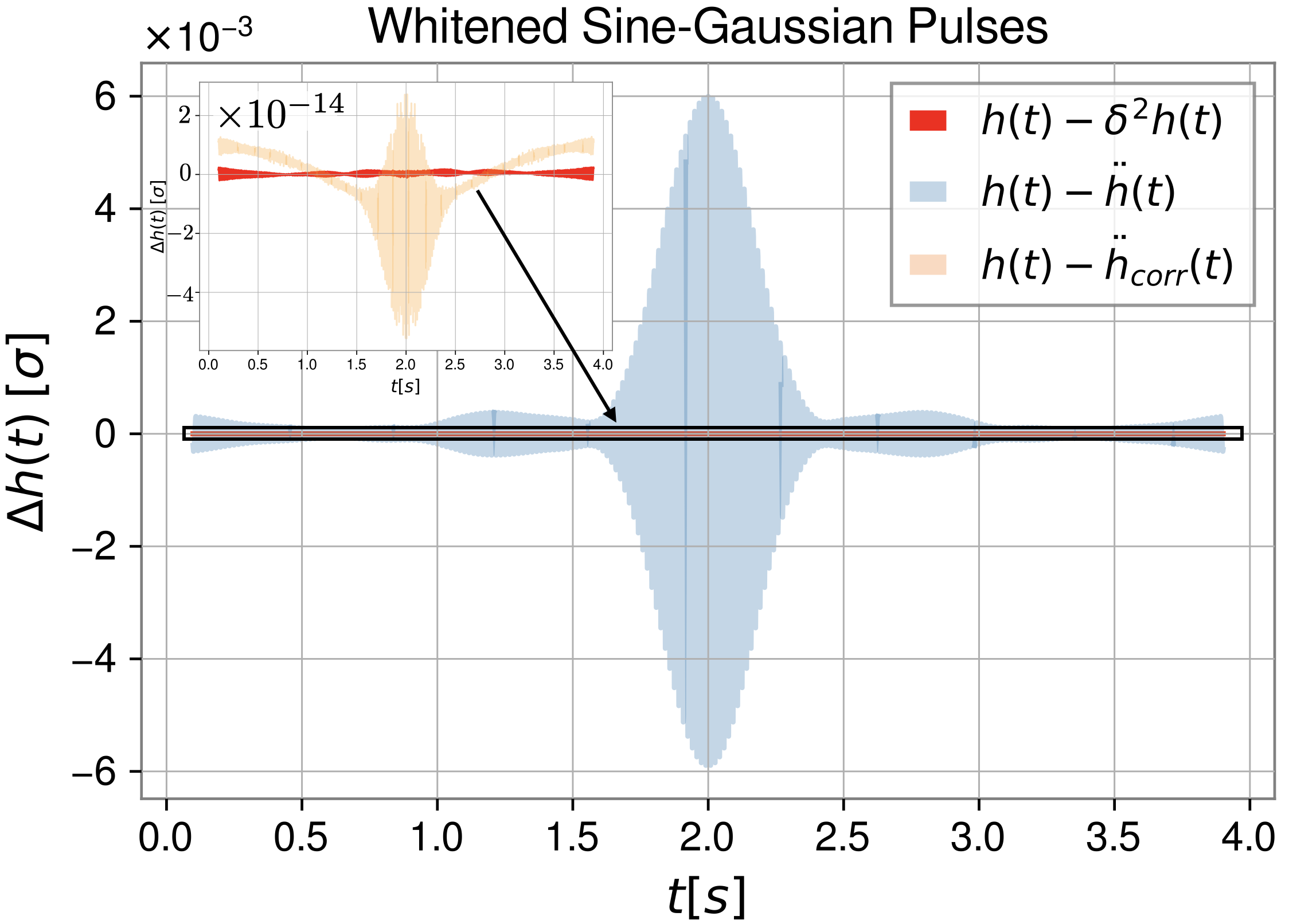}
        \caption{\textbf{Demonstration of our transformation and whitening scheme on sine-Gaussian pulses.} The left panel shows the analytical second derivative $\ddot{h}(t)$ of a sine-Gaussian strain time-series $h(t)$ and its second-order finite difference time-series $\delta^2 h(t)$. We obtain the latter both directly and from correcting $\ddot{h}(t)$ via Eq.~\eqref{eq:psiPsi}. The inset of the right panel shows the difference between the latter two time-series, whitened with a PSD $S_{n_{\Psi_4}}$, and the original strain $h(t)$ whitened by the corresponding $S_{n}$. These are of the order of 1 part in $10^{12}$. The main panel shows the (much larger) difference between the whitened strain and second derivative $\ddot{h}(t)$ whitened with $S_{n_{\Psi_4}}$.}
        \label{fig:whitening_sg}
    \end{center}
\end{figure*}

\begin{figure*}
    \begin{center}
        \includegraphics[width=0.49\textwidth]{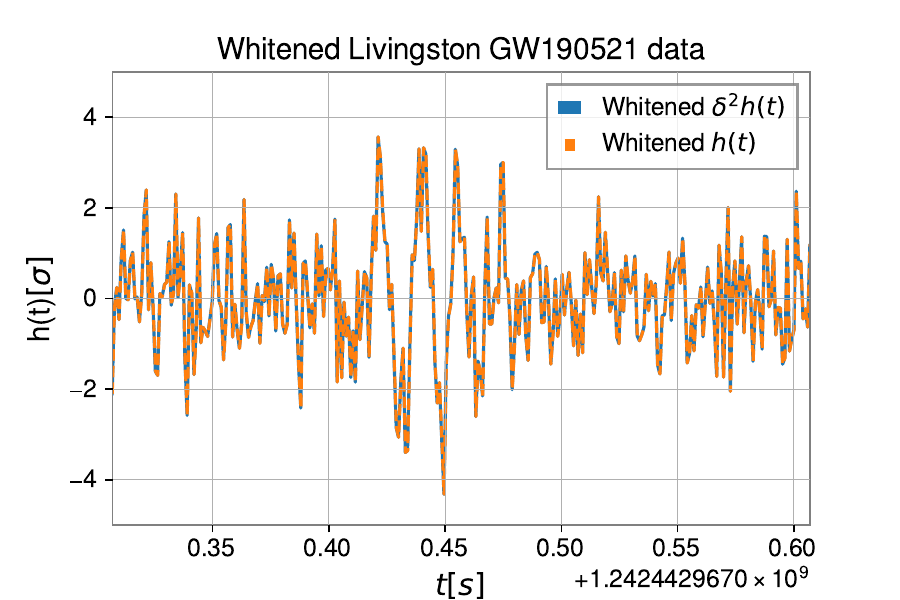}
        \includegraphics[width=0.49\textwidth]{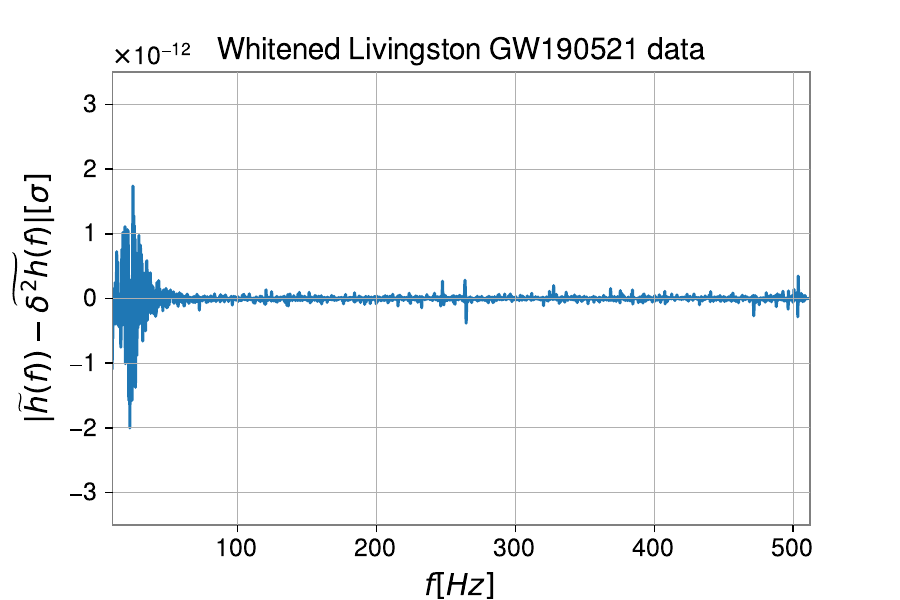}
        \caption{\textbf{Whitening of strain and $\Psi_4$ detector data.} Left: whitened strain and $\Psi_4$ gravitational-wave time-series from the Livingston detector around the time of GW190521. Right: difference between the absolute values of the corresponding Fourier-domain data. These are at the level of one part in $10^{12}$, so that both whitened detector data are equivalent for all practical purposes.}
        \label{fig:whitening}
    \end{center}
\end{figure*}

\subsection{Analytic case: Sine-Gaussian waveform}

We start by considering the case of an analytical function $h(t)$ for which we can analytically compute $\psi_4(t)=d^2h(t)/dt^2 \equiv \ddot{h}(t)$. This way, we have a ``controlled'' experiment where the ``strain'' $h(t)$ and the corresponding ``Newman-Penrose scalar'' $\psi_4(t)$ at hand are free of potential differences introduced by systematic errors arising from the double integration of the former. In particular, we consider the case of a sine-Gaussian strain time series
\begin{equation}
    h(t)=A_0\,e^{-(t-t_0)^2/\tau}\cos(\omega t + \phi_0).
\end{equation}
For this strain, we compute the corresponding $\psi_4(t)$ and obtain a corresponding finite sampling time-series $\psi_4[m]$. Next, we obtain a finite sampling-time series of the strain $h[m]$ and compute the corresponding compute second-order difference $\Psi_4[m]\equiv (\delta^{2}h)[m]$. Finally, we also obtain $\delta^{2}h[m]$ by ``correcting'' the discretized second derivative $\ddot{h}[m]$, via the transformation in Eq.~\eqref{eq:psiPsi}. 
The left panel of Fig.~\ref{fig:whitening_sg} shows these three time series. The inset therein shows that while the two finite-differenced time series are identical, $\ddot{h}(t)$ differs from them. \\

Next, we compare the result of whitening the strain time-series $h[m]$ by a given strain PSD $S_n$ and that of whitening $\delta^{2}h[m]$ by the corresponding transformed PSD $S_{n_{\Psi_4}}$. The inset of the right panel of Fig.~\ref{fig:whitening_sg} shows the difference between the whitened Fourier transforms of $h[m]$ and $\delta^{2}h[m]$. As before,  we obtain the latter both by taking second-order finite differences on $h[m]$ and by correcting $\ddot{h}[m]$ via Eq.~\eqref{eq:psiPsi}, which we denote $\ddot{h}_\text{corr}$. These differences are below 1 part in $10^{12}$. The main panel shows the differences between the whitened $h[m]$ and $\ddot{h}[m]$, ``wrongly'' whitened by $S_{n_{\Psi_4}}$, which are nine orders of magnitude larger.\\

The above shows that given a continuous strain $h(t)$ and its corresponding $\psi_4(t)\equiv \ddot{h}(t)$, the processes of a) taking second-order finite differences on the finite-sampling time-series of $h[m]$ and b) correcting $\ddot{h}[m]$  via Eq.~\eqref{eq:psiPsi} lead to identical time series that we call $\Psi_4(t)$. Second, it shows that given the PSD of a stochastic Gaussian and stationary strain time-series $h[m]$ and the corresponding $(\delta^2)h[m] = \Psi_4[m]$, our estimation of the PSD $S_{n_{\Psi_4}}$ correctly whitens the latter. In the following, in order to adapt to common literature, we drop the discrete notation, e.g. replacing $h[m]$ by $h(t)$.

\subsection{Whitening of detector data}

The left panel of Fig.~\ref{fig:whitening} shows the whitened strain $d(t)$ and $d_{\Psi_4}(t)$ time-series of the Livingston detector at the time of the event GW190521. Their differences, shown in the right panel, are well below one part in $10^{12}$. Again, this shows that our formalism correctly whitens the data and that, therefore, both types of analyses are totally equivalent provided that no artefacts are picked during the construction of the strain templates from the $\psi_4$ ones. 

\subsubsection{Whitening waveform templates:\\ impact of the choice of $\omega_0$ during fixed-frequency integration}

The left panel of Figs.~\ref{fig:whitening_21gtemplate-1} and~\ref{fig:whitening_14ftemplate} show raw strain templates two simulations of a Proca star merger $h(t)$.
These respectively correspond to a waveform consistent with GW190521 and to a larger mass-ratio and rather edge-on configuration with multi-modal structure \cite{Bustillo:2016gid,CalderonBustillo:2019wwe,Graff:2015bba} that, in a separate paper \cite{psi4_observations}, we find consistent with the GW trigger 200114$\_$020818 (S200114f in the following) \cite{O3IMBH} \footnote{Please see the specific parameters in Appendix C}. In both, cases, the strain $h(t)$ has been obtained from $\psi_4(t)$ through an FFI using a given $\omega_0$ cutoff. Overlaid, we show the corresponding $\Psi_4(t)$ obtained both as $(\delta^{2}h)(t)$ and by correcting $\psi_4(t)$ outputted by NR, which in the following we simply call $\Psi_4(t)$. We scale $h(t)$ by a suitable amplitude factor so that both waveforms can be plotted together. The right panel shows the corresponding whitened waveforms.

First, we note that while (as expected) the raw $h(t)$ widely differs from the two $\Psi_4(t)$ waveforms, the whitened $h(t)$ and $(\delta^{2}h)(t)$ waveforms are identical but differ from the ``direct'' $\Psi_4$. This is due to the impact of the choice of $\omega_0$ used to obtain $h(t)$ from $\psi_4(t)$. As we will show later, for the case shown in Fig.~\ref{fig:whitening_21gtemplate-1}, these differences are not large enough to have a significant impact on parameter inference or model selection. However, for the case shown in Fig.~\ref{fig:whitening_14ftemplate}, the choice of $\omega_0$ removes enough ``true'' signal power to cause clear morphological alterations that impact both parameter estimation and model selection.

\begin{table*}
    \centering
    \begin{tabular}{cc|rlrlrl|rlrlrl}
        \multicolumn{2}{c|}{Waveform model} & \multicolumn{6}{c|}{GW190521-like injection} &  \multicolumn{6}{c}{S200114f-like injection}\\
        \hline
        \rule{0pt}{3ex}%
         & & \multicolumn{2}{c}{SNR = 15} & \multicolumn{2}{c}{SNR = 30} & \multicolumn{2}{c|}{SNR = 60} & \multicolumn{2}{c}{SNR = 15} &  \multicolumn{2}{c}{SNR = 30} & \multicolumn{2}{c}{SNR = 60} \\ 
        \rule{0pt}{3ex}%
        Injection $M$ & Template $M^*$ & $\log{\cal B}$ & $\log{\cal L}_{\rm max}$ &  $\log{\cal B}$ & $\log{\cal L}_{\rm max}$ &  $\log{\cal B}$ & $\log{\cal L}_{\rm max}$ & $\log{\cal B}$   & $\log{\cal L}_{\rm max}$  & $\log{\cal B}$   & $\log{\cal L}_{\rm max}$  & $\log{\cal B}$  & $\log{\cal L}_{\rm max}$ \\
        \rule{0pt}{3ex}%
        $\Psi_4$ & $\Psi_4$ & 94.1  & 123.2  & 477.2  & 514.3  & 2033.7  & 2063.8  & 90.0  & 124.2 & 475.9 & 517.8 & 2042.7 & 2074.4\\
        \rule{0pt}{3ex}%
        $\Psi_4$ & $\delta^{2}h_{\text{NF}}$ & 93.9  & 123.2  & 477.0  & 514.4  & 2033.4  & 2063.7  & 89.9  & 124.0 &  475.7 & 517.6 & 2041.7 & 2073.4 \\
        \rule{0pt}{3ex}%
        $\delta^{2}h_{\text{NF}}$ & $\delta^{2}h_{\text{NF}}$ & 93.8  & 123.1  & 476.5  & 513.7  & 2030.8  & 2061.3  & 89.2  & 123.9 & 475.0 & 516.8 & 2038.1 & 2070.1 \\
        \rule{0pt}{3ex}%
        $h_{\text{NF}}$ & $h_{\text{NF}}$ &   93.6  & 123.8 &  476.5 & 513.7 &  2031.0  & 2061.3 & \ncor{89.2}  & 123.1  & \ncor{474.5} & 516.8 & 2038.3 & 2070.1\\
        \rule{0pt}{3ex}%
        $\Psi_4$ & $\delta^{2}h_\text{F}$ & 93.6  & 122.9  &  475.7  & 512.7  &  2027.0  & 2057.3  & 83.9  & 117.8 & 451.0 & 492.3 & 1940.5 & 1971.9\\
        \rule{0pt}{3ex}%
        $\delta^{2}h_\text{F}$ & $\delta^{2}h_\text{F}$ & 92.4  & 121.5  & 470.7  & 507.3  & 2005.1  & 2035.5  & 64.1  & 98.2 & 372.7 & 414.1 & 1634.5 & 1665.8\\
        \rule{0pt}{3ex}%
        $h_\text{F}$ & $h_\text{F}$ & 92.2  & 121.5  & 470.1  & 507.3  & 2005.3  & 2035.5  &  64.1 & 98.3 & 372.7 &  414.1 & 1634.1 & 1665.7\\
        \rule{0pt}{3ex}%

        %Best Parameters  &   &  \\
        %\hline
        %\rule{0pt}{3ex}%
        %Quasi-circular Binary Black Hole & 85.0 &  105.2 \\
        %\rule{0pt}{3ex}%
        %Head-on Equal-mass Proca Star & 85.4 &  106.7\\ 
        %\rule{0pt}{3ex}%
        %Head-on Unequal-mass Proca Star & 86.8 &  106.5 \\ 
        %\rule{0pt}{3ex}%
        %Head-on Binary Black Hole & 79.8 &  103.2 \\ \hline
    \end{tabular}
    \caption{\textbf{Summary of injection recovery with different waveform models} We report the log Bayes factor (for model $M^{*}$ vs. noise hypotheses) obtained from our different waveform models, together with the corresponding maximum log likelihood values. We show results for two types of injections \ncor{of Proca-star merger signals}, respectively consistent with the GW190521 signal and with the 200114f trigger, both with SNRs around 15. To show the  increasing impact of $\psi_4$ integration errors as the SNR raises, we further scale our injections by factors of 2 and 4, corresponding to SNRs of approximately 30 and 60. \ncor{Log Bayes' Factors have typical uncertainties of $\simeq 0.1.$ with maximum values of $0.5$.}}
    \label{tab:injections}
\end{table*}

\begin{figure*}
    \includegraphics[width=0.50\textwidth]{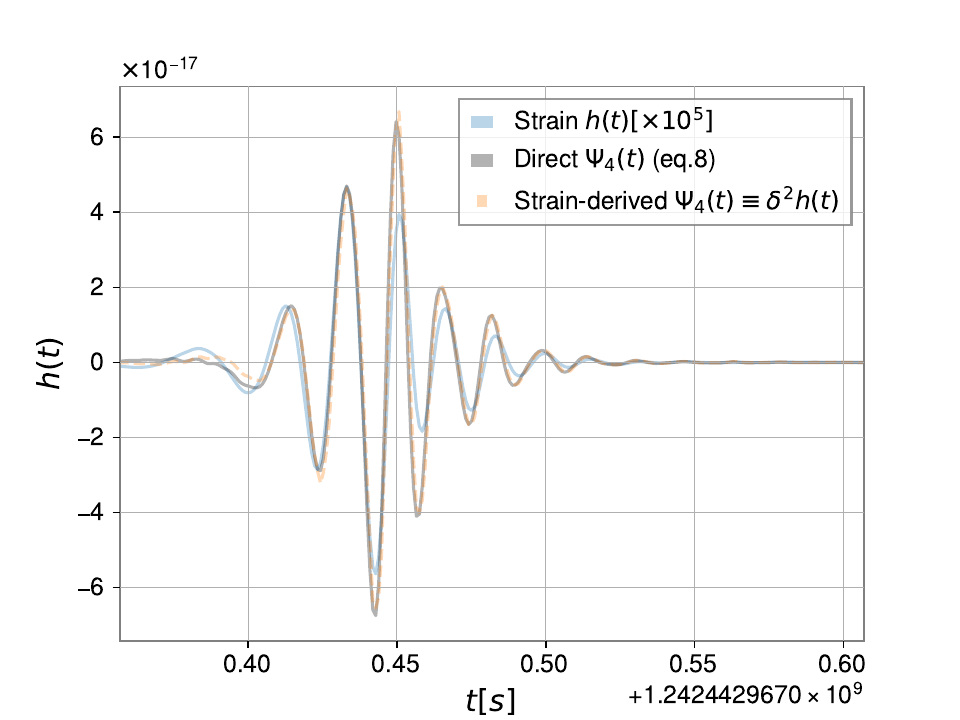}
    \includegraphics[width=0.47\textwidth]{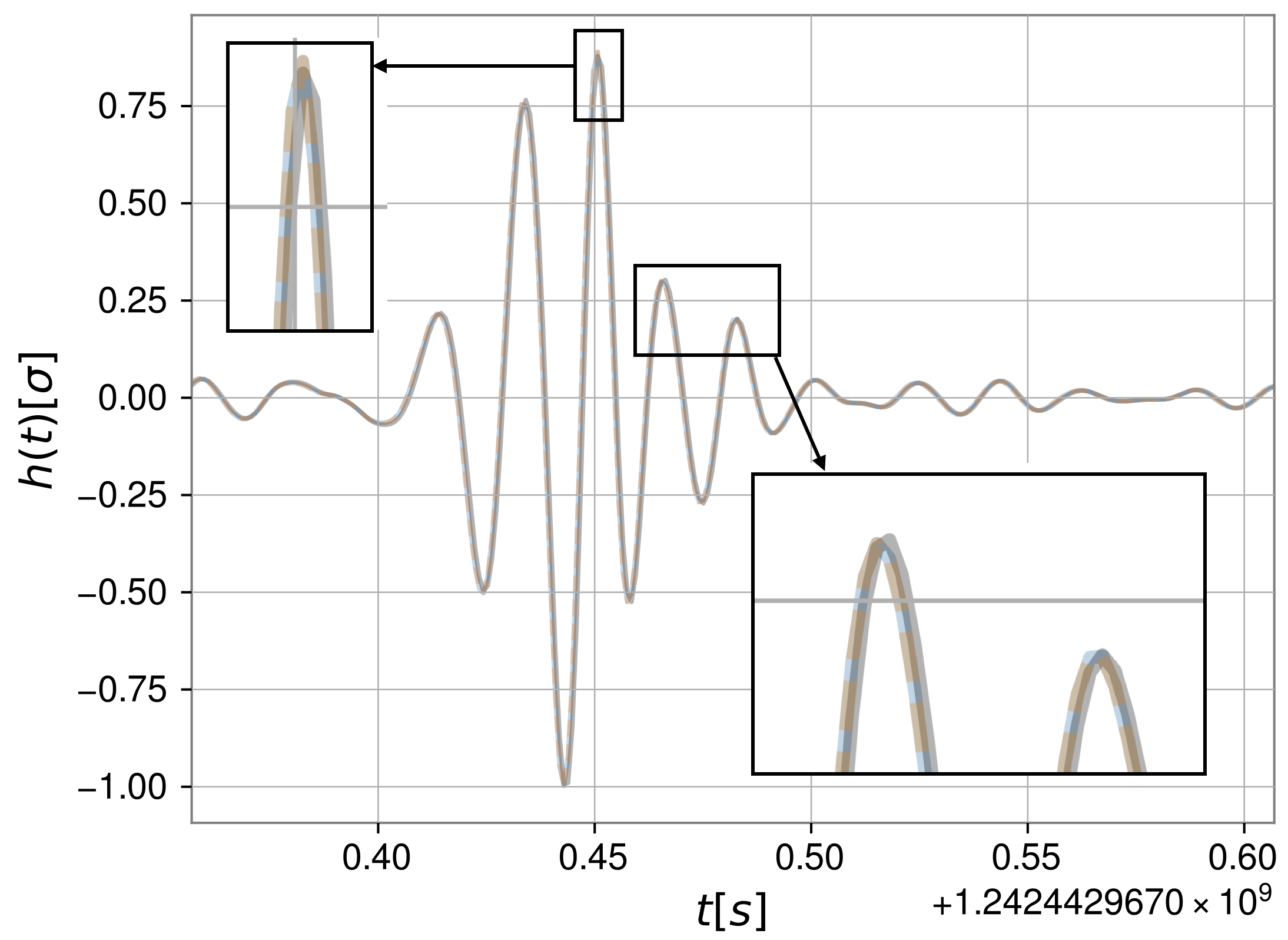}
    \caption{
    \textbf{Whitening of strain and $\Psi_4$ templates.} \textbf{Left:} We show the raw time-domain data for the case of a) $\Psi_4$ directly coming from a numerical relativity simulation (through Eq.~\eqref{eq:psiPsi}) of a head-on Proca-star merger consistent with GW190521 (black), b) the strain obtained from $\psi_4$ through double integration (blue) and c) the $\Psi_4$ obtained from the latter through second-order finite differencing, denoted by $\delta^2 h(t)$. The strain in the left panel has been conveniently scaled to note the obvious morphological differences with respect to $\Psi_4$. \textbf{Right:} corresponding whitened time-series. The zoomed boxes show how the $h(t)$ and $\delta^2 h(t)$ are exactly identical while very small differences can be observed with respect to the original $\Psi_4$.}
    \label{fig:whitening_21gtemplate-1}
\end{figure*}

\begin{figure*}
    \centering
    \includegraphics[width=0.50\textwidth]{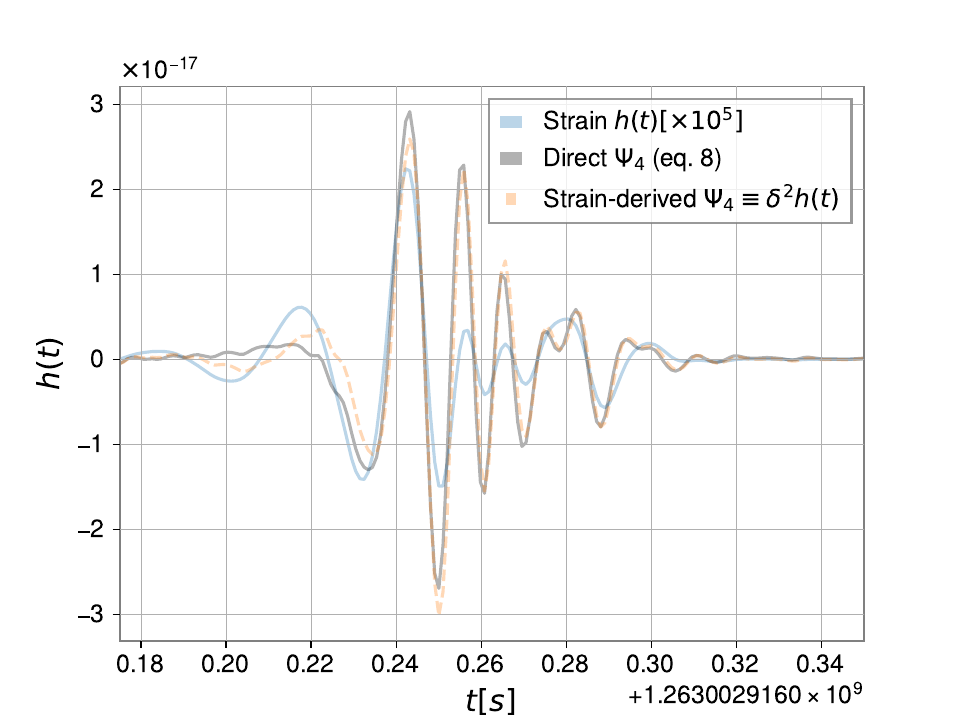}
    \includegraphics[width=0.47\textwidth]{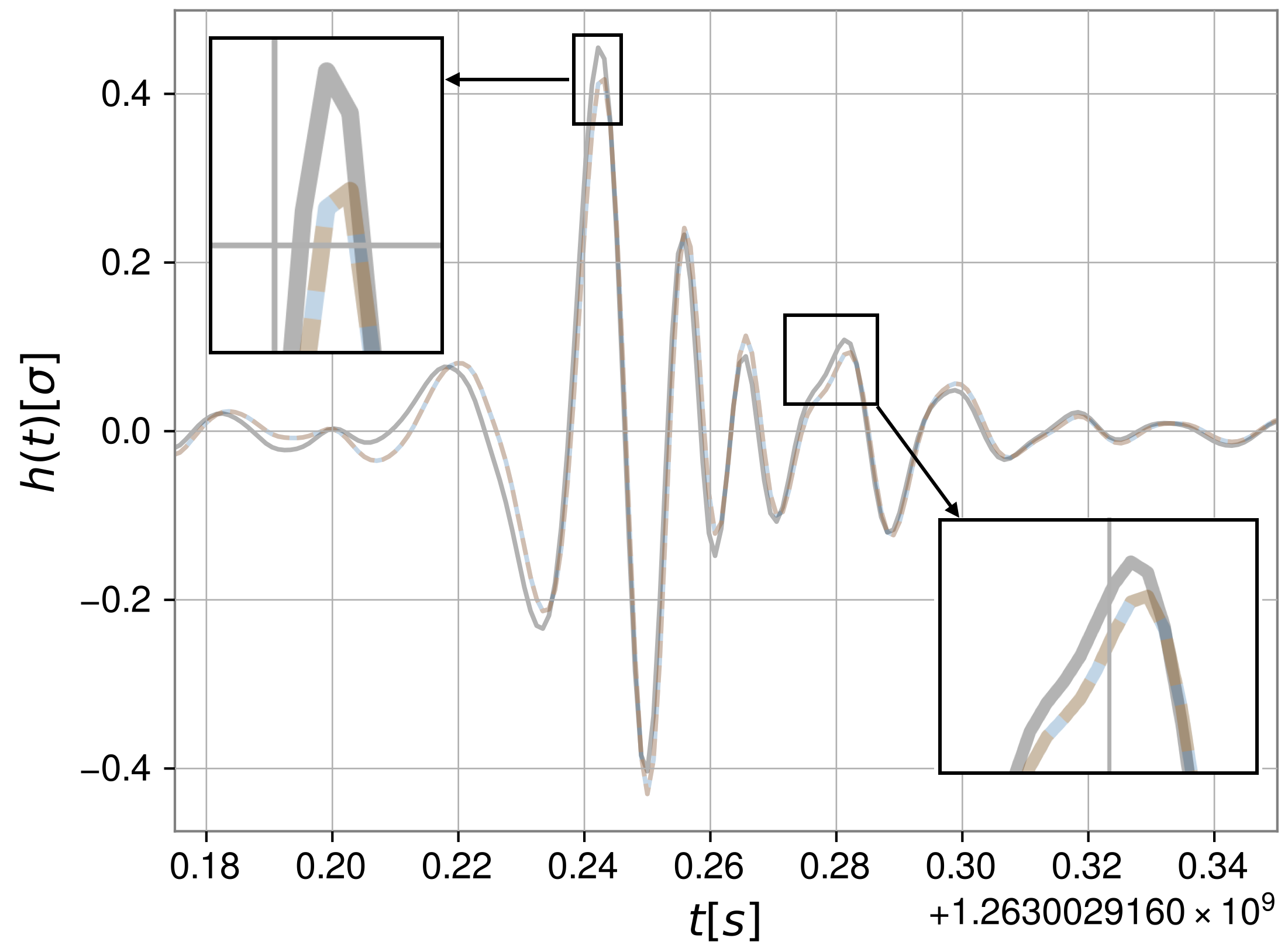}
    \caption{
    \textbf{Whitening of strain and $\Psi_4$ templates. Impact of aggressive choice of $\omega_0$.} Same as in Fig.~\ref{fig:whitening_21gtemplate-1} but for a waveform template consistent with S200114f \cite{O3IMBH,psi4_observations}. In this case, the differences between the $\Psi_4$ directly extracted from the numerical simulation and the other two waveforms are significantly more noticeable.}
    \label{fig:whitening_14ftemplate}
\end{figure*}

\section{Parameter inference and model selection on simulated signals}

\subsection{Summary of Bayesian Parameter Inference and Model selection}

We test our framework by performing full Bayesian parameter estimation and model selection on simulated signals injected in zero-noise using the Bayesian inference library \texttt{Parallel Bilby}. We consider a reference signal or ``injection'' $h_M(\theta_{\texttt{True}})$ with source parameters $\theta_{\texttt{True}}$ computed by a waveform model $M$. In our case, the model $M$ corresponds to either $\Psi_4(t)$, $h(t)$ or $(\delta^{2}h)(t)$. Next, we recover the posterior distributions of the parameters $p_{M^{*}}(\theta\,|\,h(\theta_{\texttt{True}}))$ using a different model $M^{*}$ as the signal template. This is given by
\begin{equation}
    p_{M^{*}}(\theta\,|\,h_{M}(\theta_{\texttt{True}})) = \frac{ {\cal{L}}_{M^{*}}(h_{M}(\theta_{\texttt{True}}) \,|\, \theta)\,\pi(\theta)}{{\cal{Z}}_{MM^{*}}}.
\end{equation}

Here, $\pi(\theta)$ denotes the prior probability of the parameters $\theta$ while ${\cal{L}}_{M^{*}}(h_{M}(\theta_{\texttt{True}}) \,|\, \theta)$ represents the likelihood of the data $h_{M}(\theta_{\texttt{True}})$ under the waveform model $M^{*}$ with the given parameters $\theta$. We use the canonical likelihood for GW transients in Eq.~\eqref{eq:logl}. Finally, the term ${\cal{Z}}_{MM^{*}}$ denotes the Bayesian evidence for the data $h_M$ assuming the template model $M^{*}$. This is equal to the integral of the numerator over the explored parameter space $\Theta$, given by  
\begin{equation}
    {\cal{Z}}_{MM^{*}}=\int_{\Theta} \pi(\theta)\,{\cal{L}}_{M^{*}}(h_{M}(\theta_{\texttt{True}}) \,|\, \theta) \,\dd\theta
\end{equation}
Given two template models $M_1$ and $M_2$ being compared to some data $d$, or some simulated signal $h_M(\theta)$, the relative probability for those models, or relative Bayes factor ${\cal{B}}^{M_1}_{M_2}$ is given by 
\begin{equation}
    {\cal{B}}^{M_1}_{M_2} = \frac{{\cal{Z}}_{M_1}}{{\cal{Z}}_{M_2}}
\end{equation}
Expressing these in terms of natural logarithms, it is commonly considered that the model $M_1$ is strongly preferred with respect to $M_2$ when $\log({\cal{B}}^{M_1}_{M_2}) = \log({\cal{Z}}_{M_1})-\log({\cal{Z}}_{M_1}) \geq 5$. Finally, since the Bayesian evidence $\cal{Z}$ represents the Bayes Factor for the ``model vs. noise'' hypotheses, we will commonly refer to it as simply the ``Bayes Factor'', denoting it as $\cal{B}$.\\

As it will become relevant later, it is important to note that, the evidence ${\cal{Z}}_{MM^{*}}$ is \ncor{bounded above} by the maximum value of the likelihood ${\cal{L}}_{M^{*}}(h_{M}(\theta_{\texttt{True}})|\theta_{\texttt{best}})$, achieved for the best fitting parameters $\theta_{\texttt{best}}$. This is, by the best fit that the model $M^{*}$ can provide for $h_M(\theta)$. At the same time, \ncor{in the absence of noise}, such maximum likelihood is capped by the ``optimal maximum likelihood'' ${\cal{L}}_{M}(h_{M}(\theta_{\texttt{True}})|\theta_{\texttt{True}})$.\\

To anticipate the expected consequences of respectively analysing and modelling a true GW using a template affected by integration errors, let us consider two scenarios. First, consider that we model the true GW, i.e. our injection, as $h_M(\theta_\texttt{True}) = \Psi_4(\theta_\texttt{True})$ and try to recover it  using templates $(\delta^2 h)(\theta_\texttt{True})$ which carry integration artefacts. Since such artefacts will change the frequency content of the templates, these will not perfectly match the injection, leading to a drop in the maximum likelihood and, therefore, of the corresponding Bayesian evidence in favour of the model. Second, any choice of the integration frequency cutoff $\omega_0$ will remove some true power from the waveform. This will consequently lead to an under-estimation of the signal loudness for a given source distance, yielding a bias toward lower distances. Second, for this last reason, if we model the true signal using either  $h(\theta_{\texttt{True}})$ or $(\delta^2 h)(\theta_\texttt{True})$, this will cause an under-estimation of the signal loudness, the optimal maximum likelihood and, therefore, an intrinsic decrease of the maximum Bayesian evidence achievable in the analysis.

\subsection{Specific set-up}

We perform parameter estimation on injections generated in terms of $\Psi_4(t)$, $h(t)$ and $(\delta^{2}h)(t)$. \textcolor{black}{For the latter two, we consider two cases. In the first case, we obtain the strain through FFI using a frequency cutoff M$\omega_0 \simeq 0.27$, in geometric units}\footnote{Within our simulation set \cite{psi4_observations}, we found that this was typically the lowest value leading to no secular drifts. As we note, however, in some cases this choice removes ``true signal power'' from the detector band. In particular, for the total masses chosen for our GW190521-like and s200114f-like injections, our cutoff frequency $M\omega_0$ translates to 32.3 and 37.2 Hz respectively.}. In the second, we simply apply a regularization at the pole given by $\omega_0=0$, which we replace by the value for the lowest frequency multiplied by $10^{-4}$. We respectively label the resulting waveforms by F and NF sub-indexes, i.e., $h_{\text{F}}$, $(\delta^{2}h_\text{F})$  and $h_{\text{NF}}$,\ $(\delta^{2}h_\text{NF})$.

We recover these injections using different types of templates, as shown in Table~\ref{tab:injections}. We make two choices for the parameters $\theta$, corresponding to the two cases shown in Figs.~ \ref{fig:whitening_21gtemplate-1} and \ref{fig:whitening_14ftemplate}. These correspond to parameters consistent with GW190521 and the trigger S200114f under our $\Psi_4$ formalism. The most relevant difference between the corresponding simulations is the aggressiveness of the $\omega_0$ used to obtain $h(t)$. 

As mentioned in the previous section, in the case of the simulation consistent with GW190521, we find that this does not subtract significant power from the portion of the signal falling into the Advanced LIGO sensitive band while in the second case (the S200114f-like simulation) it does. The expectation is that for the first case, results obtained through $\Psi_4$ and all $h_{\text{F}}$-based analyses will be very similar; while in the second, those based on filtered waveforms will differ significantly. In particular, two types of differences are expected. First, if the frequency content of the waveforms is altered by the integration errors, this will limit the ability of the resulting strain waveforms (or rather, $(\delta^{2}h)(t)$) to fit the original $\Psi_4$. This will translate into both a reduction of the Bayes Factor that may bias model selection and into potential parameter biases. Second, since any choice of $\omega_0$ will remove a certain amount of signal power, this will result in an underrating of the strain-signal loudness. On the one hand, for identical parameters, this will lead to an under-estimation of the signal SNR. On the other hand, this will cause a bias in the distance estimate. 

The significance of the above effects in model selection and parameter inference depends on the signal loudness, as louder signals require more accurate templates in order to avoid analysis biases. We evaluate the impact of these biases under various observing scenarios, we consider three types of signal loudness, characterised by the optimal signal-to-noise ratio (SNR) of the injection modeled by $\Psi_4$. In the first case, where use the exact parameters best-fitting GW190521 and s200114f, the $\Psi_4$ injection has an SNR of $\simeq 15$ across the whole detector network, typical of current GW detections. Next, reducing the distance by a factor of 2, we study the case of signals with SNR $\simeq 30$, similar to the maximum SNR observed to date. Finally, we consider the case where the injection has an SNR of $\simeq 60$. \ncor{For simplicity, in what follows, we will use ``$=$'' signs to refer to these cases}. 

Finally, as shown in the previous section, if our $\Psi_4$ formalism is equivalent to the classical one based on strain, results obtained through the injection and recovery of $(\delta^{2}h_{\text{F/NF}})(t)$ and $h_{\text{F/NF}}(t)$ should be exactly equal, modulo the uncertainty associated to the sampling of the likelihood throughout the parameter space. \ncor{In fact, Table \ref{tab:injections} shows that the evidences obtained by such pairs of analyses differ at most by 0.5 (which would not impact our conclusions regarding model selection), even in the highest SNR cases where convergence is harder to achieve, with most cases ranging between 0 and 0.2.}

When sampling the likelihood, we fix the mass-ratio and spins of the templates to those of the injection and sampling only over the total red-shifted mass of the source $M_{\text{total}}$, the luminosity distance $d_{\rm L}$ and orientation angles $(\iota,\varphi)$, the sky-location angles $(\alpha,\delta)$, the polarization angle $\psi$ and the time-of-arrival. The power spectral densities used for our two injections are those of the Advanced LIGO and Virgo detectors at the times of GW190521 and S200114f. When analysing the corresponding $\Psi_4$ or $\delta^{2}h$ injections, we applied the correction factor in Eq.~\eqref{eq:psi4_psd} to obtain the appropriate PSDs. We sample the parameter space using the nested sampler \texttt{Dynesty} \cite{Dynesty} with 4096 live points for the cases with SNR=15 and 30, and 8192 live points for the cases with SNR = 60. 

%% For the following 3 figures, remove "-2" to revert to previous colour scheme.
\begin{figure*}
    \centering
    %% Specifying picture height to ensure both PE results can be shown on the same page, 
    %% each occupying half a page.
    \includegraphics[height=.4\textheight]{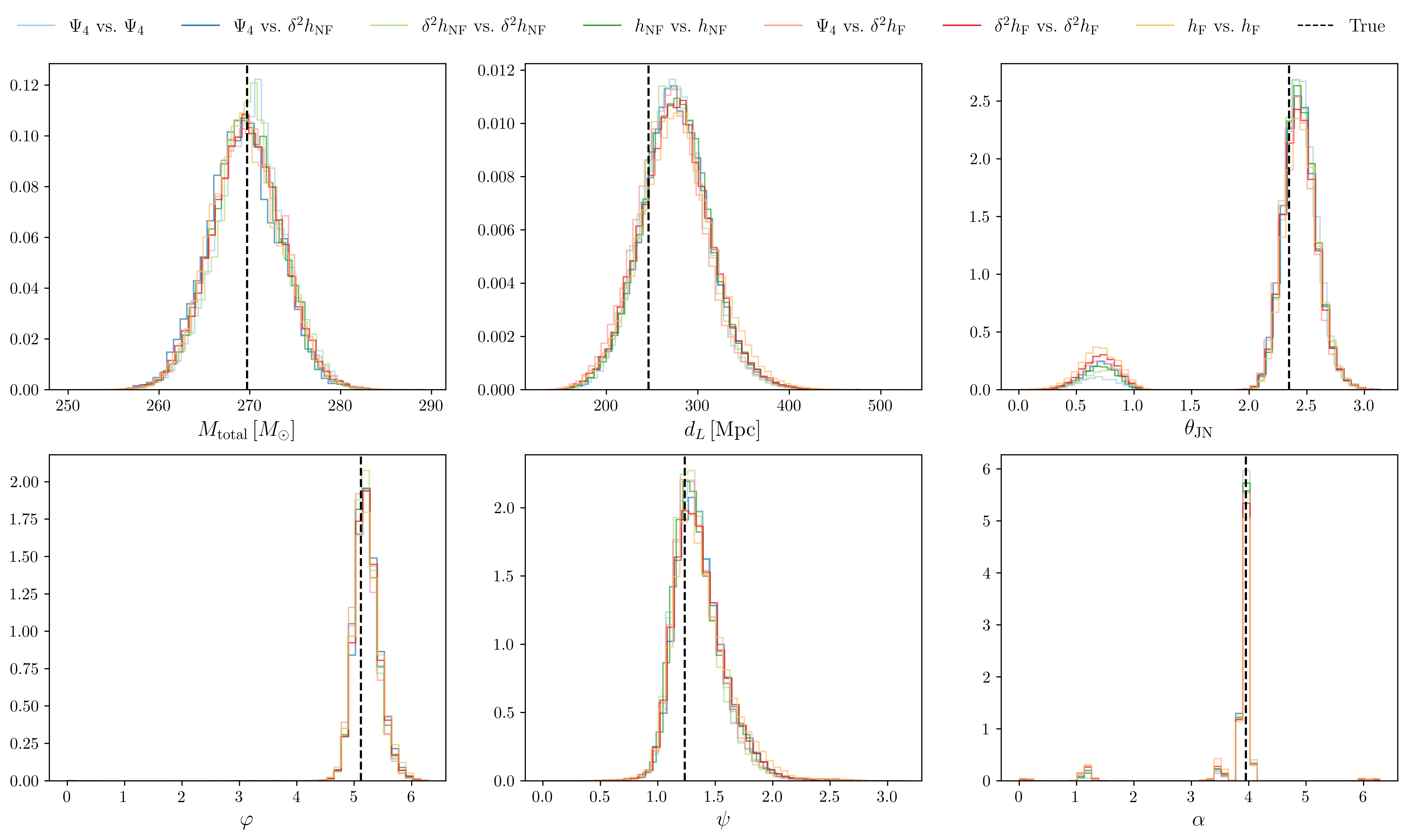}
    \caption{
    \textbf{Posterior parameter distributions for our GW190521-like injection, when scaled to a signal-to-noise ratio of 15} Posterior parameter distributions for our different analyses in Table \ref{tab:injections} together with the true value represented by a dashed line. The color code denotes the type of injection used ($\Psi_4$, strain $h(t)$ or strain-derived $\Psi_4$ denoted by $\delta^2 h(t)$), and the type of template. All analyses yield equivalent results. In particular, no significant difference is observed when filtered or non-filtered injections and templates are used. The parameter $\theta_{JN}$ describes the angle formed between the total angular momentum of the source and the line-of-sight. We note that since our sources do not precess, this is equal to the parameter $\iota$.}
    \label{fig:21gPE}
\end{figure*}

\begin{figure*}
    \centering
    %% Specifying picture height to ensure both PE results can be shown on the same page, 
    %% each occupying half a page.
    \includegraphics[height=.4\textheight]{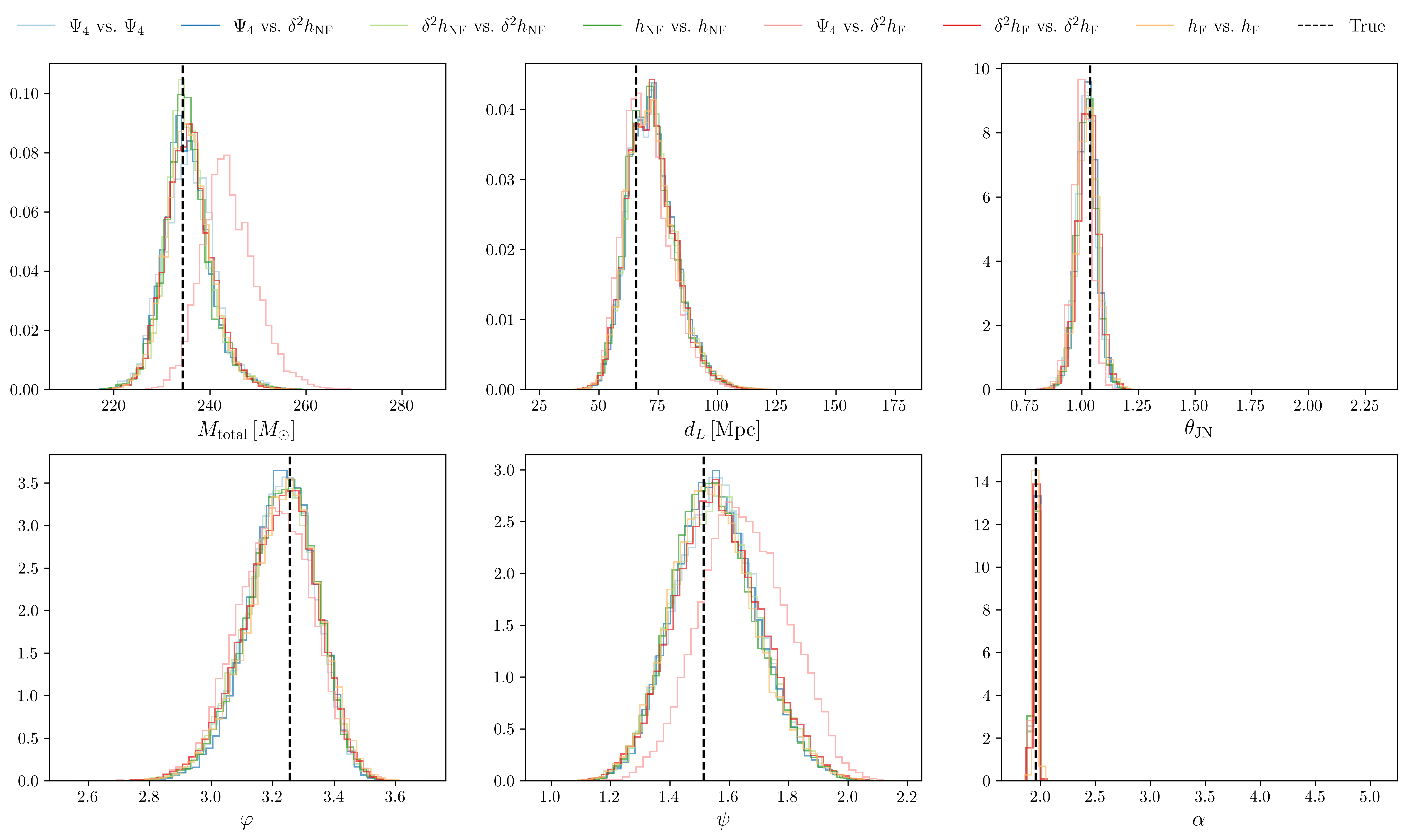}
    \caption{
    \textbf{Posterior parameter distributions for our S200114f-like injection, when scaled to a signal-to-noise ratio of 15} Posterior parameter distributions for our different analyses in Table \ref{tab:injections} together with the true value represented by a dashed line. The color code denotes the type of injection used ($\Psi_4$, strain $h(t)$ or strain-derived $\Psi_4 \equiv (\delta^2 h)(t)$) and the type of template. Recovering non-filtered injections with filtered waveforms leads to visible shifts in some posteriors equivalent results. This is due to the excessive aggressiveness of the integration filter.}
    \label{fig:14fPE}
\end{figure*}

\begin{figure}
    \begin{center}
        \includegraphics[width=0.49\textwidth]{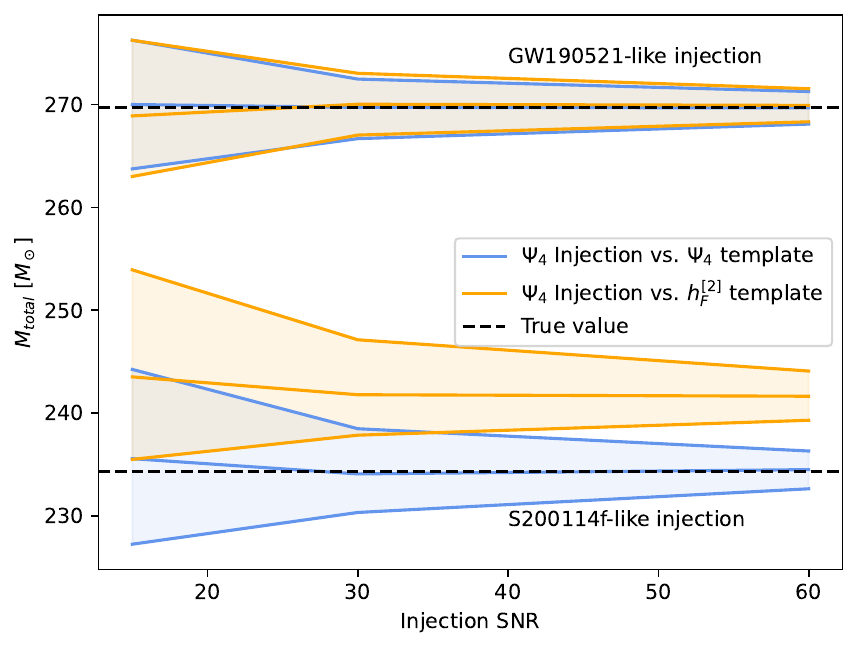}
        \caption{
        \textbf{Total mass bias due to  $\psi_4$-integration and filtering as function of signal loudness} The blue and orange contours denote the $90\%$ credible intervals for the total mass for the case of our GW190521-like and S200114f-like injections as a function of the injection signal-to-noise ratio. The injection is always modelled by $\Psi_4$, free of integration errors. The blue and orange contours denote, respectively, the result of analysing the injection with $\Psi_4$ itself and with the $(\delta^{2}h_F)$, which inherits the $\psi_4$-integration errors together with the power loss due to the choice of an integration low-frequency cutoff $\omega_0$.}.
        \label{fig:pe_snr}
    \end{center}
\end{figure}

\begin{figure*}
    \begin{center}
        \includegraphics[width=0.95\textwidth]{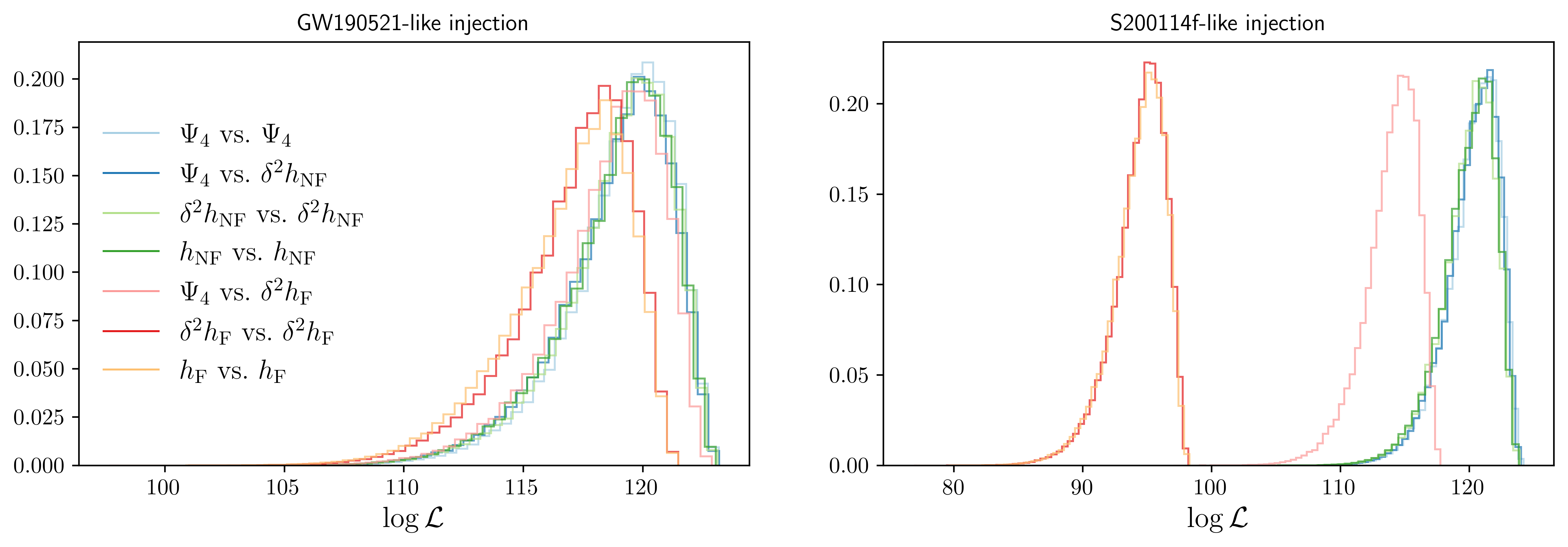}
        \caption{
        \textbf{Likelihood posterior distributions for our two sets of parameter inference runs} Left: Posterior distributions for the case of our GW190521-like injections. Right: same for our S200114f-like injections}.
        \label{fig:logl}
    \end{center}
\end{figure*}

\subsection{Results on simulated signals}

Table \ref{tab:injections} shows the natural log Bayes factor ($\log\cal{B}$) and the maximum log likelihood ($\log{\cal{L}}_{\text{max}}$) recovered by each of our analyses for each of our two types of injections. First, we note that in both cases, the analyses making use of strain waveforms $h_\text{F(NF)}$ and the corresponding second-order finite differenced waveforms $\delta^2 h_\text{F(NF)}$ yield equivalent results even for SNRs of 60, corroborating that our formalism does not introduce any artefacts. This means \textit{there is no fundamental reason to prefer an analysis based on strain}. Therefore, given that the $\Psi_4$ waveforms, directly obtained from $\psi_4$, avoid a complete layer of systematic errors, in the following we use the results based on the injection and recovery of $\Psi_4$ itself as our reference results. \ncor{Since we checked above that sampling errors introduce maximal uncertainties in the Bayes Factors of $\simeq 0.5$ (irrelevant for the purpose of model selection), we will assume that any significant difference between our reference analysis and the remaining ones are due to artefacts arising from the integration of $\psi_4(t)$ to obtain $h(t)$.}

Fig.~\ref{fig:21gPE} shows posterior parameter distributions for all injection-template combinations in Table~\ref{tab:injections}, for the case of our GW190521-like signal scaled to an SNR of 15 . The vertical bars show the true injection values. Similarly, Fig.~\ref{fig:14fPE} shows the same for our S200114f-like injection. In the first case, all distributions yield equivalent results. In particular, since our choice of $\omega_0$ barely affects the signal morphology, there is no significant difference between the posteriors obtained when injecting either $\Psi_4(t)$, $h_{\text{F}}(t)$ or $h_{\text{NF}}(t)$ and recovering with any of the relevant template models. We find the same is true when we raise the SNR to 30 and 60. In particular, the top contours of Fig. \ref{fig:pe_snr} show the symmetric $90\%$ credible intervals around the median obtained for the total mass as a function of the SNR of the $\Psi_4$ injection when this is recovered with $\Psi_4$ itself (blue) and $\delta h^{2}_\text{F}$, which carries potential integration and $\omega_0$-choice artefacts. Both these contours are essentially equal and converge to the true value for increasing SNR, indicating that, in this case, the aforementioned artefacts are mild enough to not to bias parameter estimation. However, the likelihoods reported in Table~\ref{tab:injections} and on the left panel of Fig.~\ref{fig:logl} show that the runs involving injections ($h_{\text{F}}(t)$ and $(\delta^2 h_{\text{F}})(t)$) reach slightly lower log-likelihoods due to the (small) power eliminated by the choice of $\omega_0$, which reduces the SNR of the injection. Nevertheless, as Table \ref{tab:injections} shows, for the cases with SNR = 15 and 30 such missing power does not cause changes in the Bayes factors that can lead to qualitatively different conclusions when performing model selection. Accordingly, we will later show that the analysis of GW190521 is not impacted at all by the usage of $h_F$ templates. This is however not the case when the SNR is raised to 60. In this situation, while parameter estimation is unaffected, we observe that analysing a true GW with our filtered waveforms causes a drop of $\simeq 6$ in the Bayesian evidence, sufficient to lead to model selection biases.

The situation is quite different for the case of our second injection. In this case, the choice of $\omega_0$ does significantly affect both the morphology and signal power of the waveform. As a consequence, Fig.~\ref{fig:14fPE} shows clear shifts of the posteriors for the total mass and the polarization angle when we recover the $\Psi_4$ injection with the $\delta^2 h_{\text{F}}$ templates \textit{at an SNR of only 15}, even if these are not completely inconsistent with the others. Figure \ref{fig:pe_snr} shows, however, that such shifts bias the estimate of the total mass when the SNR is above 30. This turns even more dramatic when evaluating the impact on the recovered log-likelihood and on model selection. The yellow distribution in the right panel of Fig.~\ref{fig:logl} shows that the ``F'' templates (e.g. $(\delta^2 h_{\text{F}})$) fail to recover a significant amount of power from $\Psi_4$, therefore dramatically reducing the maximum log likelihood. This leads to a significant drop in the Bayes factor, as shown in Table~\ref{tab:injections}, that can change qualitative conclusions concerning model selection even when the SNR is only of 15. In fact, as we will show later, using $h_F$ templates has strong consequences for the analysis of S200114f. Finally, as expected, injecting any of the ``F'' waveforms leads to a significant drop in the power present in the injection and, therefore, in the recovered power and Bayes factor.   

\section{Results on real data I: GW190521 as a boson-star merger}

We now demonstrate our framework on real GW data. In Ref.~\cite{Bustillo:2021proca1} we performed parameter estimation and model selection on 4 seconds of data around the time of GW190521, comparing this event to a ``vanilla'' quasi-circular BBH model employed by the LIGO-Virgo-KAGRA (LVK) collaboration \cite{GW190521D} and to a set of numerical simulations for Proca star mergers. Here we reproduce this analysis both using the classical strain formalism and our new $\psi_4$-based framework. We obtain our $\Psi_4$ input data by applying the transformations in Eqs.~\eqref{eq:second_order_finite_difference} and \eqref{eq:psi4_psd} to the public strain data and the corresponding strain PSDs. 

We compare GW190521 to a family of numerical simulations for head-on mergers of Proca stars with equal-mass and spin. The spin of the Proca stars can be directly mapped onto the bosonic field frequency, which is uniformly distributed in $\omega/\muB\in[0.80,0.93]$ with a resolution of $\Delta\omega/\muB=0.0025$. In addition, as in \cite{Bustillo:2021proca1}, we use a secondary exploratory family of unequal-mass mergers in which the frequency of one of the stars is fixed to $\omega_1/\muB=0.895$ and the other varies uniformly in $\omega_2/\muB\in[0.80,0.93]$. We perform model selection w.r.t. the classical circular BBH case for which we choose the waveform model \nrsur~\cite{NRSur7dq4} implemented in the \texttt{LALSuite} library \cite{lalsuite}. This model includes all gravitational-wave modes with $\ell \leq 4$ and is directly calibrated to precessing NR simulations with mass ratio $q \ncor{\;= m_1/m_2} \in [1,4]$ and individual spin magnitudes $a_i\in[0,0.8]$. Moreover, the model can be extrapolated to $q=6$ and $a_i=0.99$. In our original study~\cite{Bustillo:2021proca1}, we made use of the parameter estimation software \texttt{Bilby}~\cite{Ashton:2018jfp} and sampled the likelihood across the parameter space using nested sampler \texttt{CPNest}~\cite{CPNest}. In this work, however, we switch to the parallelizable version of \texttt{Bilby}, known as \texttt{Parallel Bilby}~\cite{pbilby}, and the sampler \texttt{Dynesty}~\cite{Dynesty}. Owing to this change in software, we also repeat our original analysis based on strain data. We impose the same parameter priors as in Ref.~\cite{Bustillo:2021proca1}, as detailed below. 

\subsubsection*{Bayesian priors}

\ncor{For the intrinsic source parameters, we consider uniform priors on the field frequency $\omega/\muB \in [0.80,0.93]$ and the total red-shifted mass $M\in[50,500]\,M_\odot$. For the extrinsic parameters, we impose a distance prior uniform in co-moving volume with $d_L\in [10,10000]$\,Mpc, flat priors on the polarisation angle and time of arrival, and isotropic priors on the source orientation and sky-location. For the BBH model, we set identical priors in all of the parameters shared with the BBS model. In addition, we set uniform priors on the dimensionless spin magnitudes and isotropic priors on their orientations. As in~\cite{Bustillo:2021proca1} we set a uniform prior on the mass ratio $q \in [1/6,1]$. Finally, we note that as in \cite{psi4_observations}, and as in the analysis of S200114f we describe later, we have tried an alternative prior uniform on the (inverse) mass-ratio $Q\in[1,6]$, and we have also tried to restrict the mass-ratio ranges to $Q\in[1,4]$ and $q\in[1/4,1]$. All of these yield evidences that differ, at most, by 0.2, therefore leading to identical conclusions regarding model selection.} 

\begin{figure*}
    \begin{center}
        \includegraphics[width=0.32\textwidth]{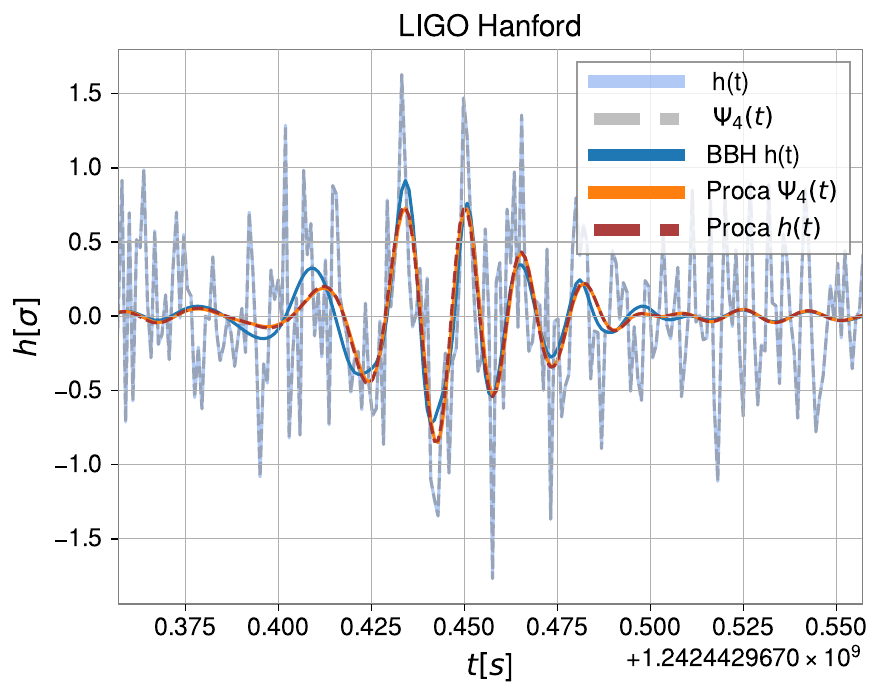}
        \includegraphics[width=0.32\textwidth]{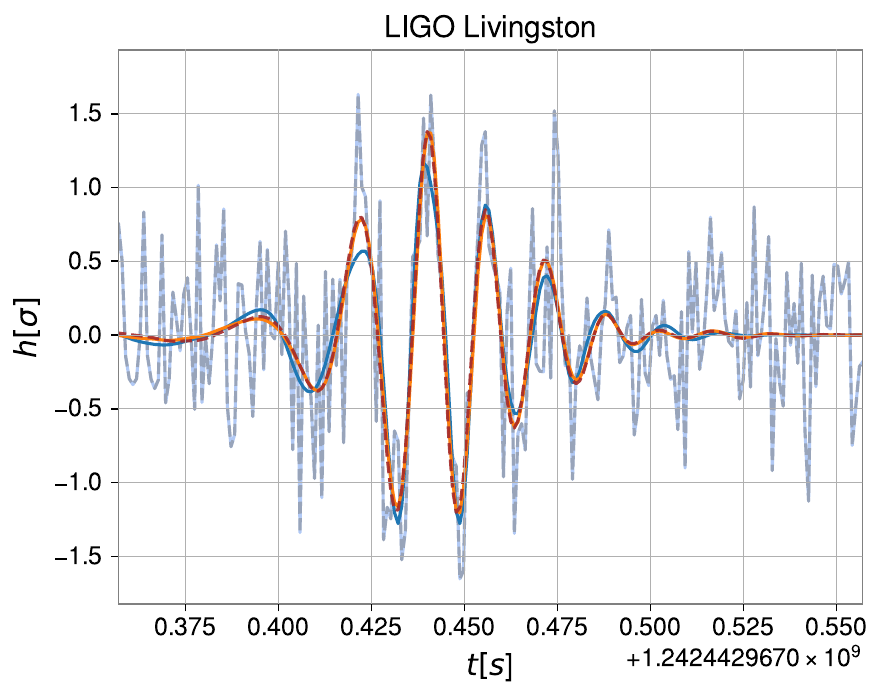}
        \includegraphics[width=0.32\textwidth]{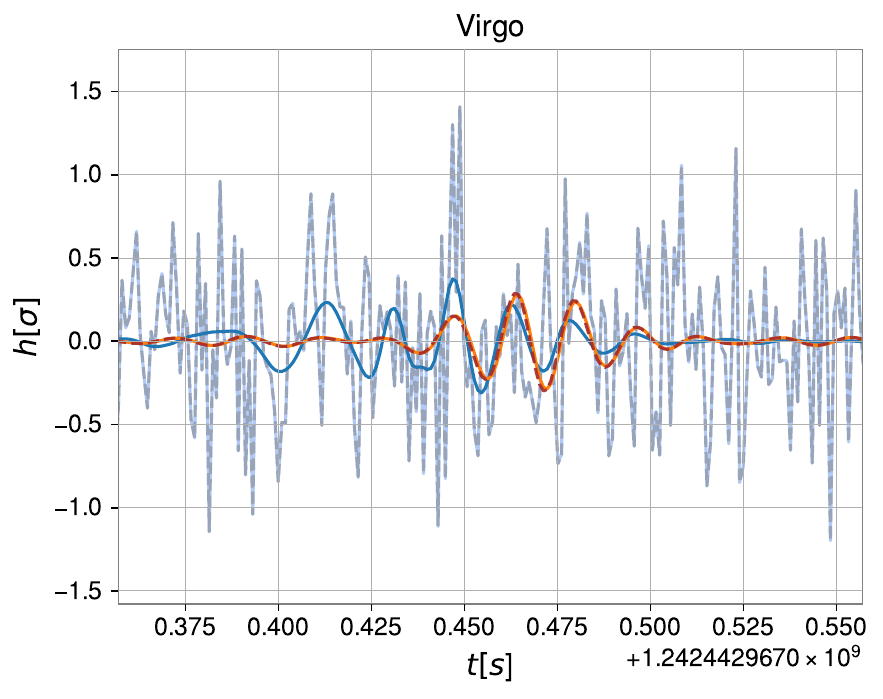}
        \caption{
        \textbf{Whitened detector data and maximum likelihood waveforms for GW190521} We show the whitened strain $h(t)$ (light blue) and $\Psi_4(t)$ (grey) detector data around the time of GW190521 together with the maximum likelihood waveforms returned by the BBH model \nrsur (blue) and by our equal-mass Proca star merger simulations obtained through strain (brown) and $\Psi_4(t)$ analyses (orange).}.
        \label{fig:maxL}
    \end{center}
\end{figure*}

\subsection{Model selection}

Fig.~\ref{fig:maxL} shows the whitened data of the Hanford, Livingston and Virgo detectors at the time of GW190521, for both the case of $h(t)$ and $\Psi_4(t)$. Together, we overlay the maximum likelihood waveforms returned by the BBH model and by both our strain and $\psi_4$-based analyses when using our equal-mass Proca star mergers. First, we note that the two latter analyses return essentially identical waveforms, once again showing that both analyses are equivalent modulo systematic errors coming from the obtention of $h(t)$. 

Table~\ref{tab:logb_21g} shows the natural logarithm of the Bayes factor ($\log\cal{B}$) for our different models under different choices of the distance prior. First, we note that for the same prior and waveform model, the $\Psi_4$ and strain analyses produce almost identical results. Second, consistently with Ref.~\cite{Bustillo:2021proca1}, the first column shows that when attaching to the standard distance prior which is uniform in co-moving volume, both the equal and unequal-mass models yield $\log\cal{B}$ only slightly larger than the BBH model. In particular, using $\Psi_4$ as our reference analysis, the  equal (unequal-mass) model is $e^{0.8}\simeq 2$ ($e^{1.7}\simeq 5.5$) times more probable than the BBH one. The second column shows results obtained under the assumption of a uniform distance prior. Although this can be considered to be rather nonphysical, this prior effectively removes the intrinsic bias towards louder sources (as circular BBHs) that can be observed from much further away than much weaker head-on mergers, which was induced by the previous prior. \ncor{Alternatively, results obtained under the uniform prior can be considered as crude estimates of what would happen once numerical simulations for (intrinsically louder) less eccentric configurations of Proca star mergers become available \footnote{Here, we make the crude assumption that, for suitable combinations of the binary parameters, we would be similar merger-ringdown signals as those obtained for our head-on mergers.}}. Once again, using our $\Psi_4$-analysis as a reference, the equal and unequal-mass models are favoured with probabilities $e^{3.3}=27:1$ and $e^{4.2}=67:1$ with respect to the BBH case. These results are perfectly consistent with those obtained from the analysis of strain data and with those reported in Ref.~\cite{Bustillo:2021proca1}.

\subsection{Parameter estimation}

Table~\ref{tab:pe} shows our parameter estimates for GW190521. We report median values together with symmetric $90\%$ credible intervals. The $q=1$ rows correspond to results obtained with our equal-mass simulations while the $q\neq 1$ rows correspond to those obtained with our exploratory unequal-mass ones. For each of these, we report results from the analysis of both strain and $\Psi_4$ data.

First, we note that, as expected from the previous section, strain and $\Psi_4$ produce almost identical results. Consistent with Ref.~\cite{Bustillo:2021proca1}, and taking $\Psi_4$-analysis using $q=1$ waveforms as a reference, we find that GW190521 can be interpreted as a head-on merger of two Proca stars with masses $117^{+5}_{-8}\,M_\odot$ that left behind a black-hole a final mass of $M_{\rm f}=233^{+12}_{-16}\,M_\odot$ and spin of $a_{\rm f}=0.70^{+0.04}_{-0.03}$; observed at a distance of $544^{+296}_{-163}$\,Mpc. Due to the much lower intrinsic loudness of head-on mergers, the inferred distance and total mass are in large contrast with those inferred by the LVK Collaboration, respectively $\simeq 5$\,Gpc and $\simeq 150\,M_\odot$. For the Proca stars, we infer a field frequency $\omega/\muB=0.895^{+0.15}_{-0.15}$. Combined with the total mass, this yields an ultralight boson mass of $8.60^{+0.63}_{-0.62}\times 10^{-13}$\,eV. Our analysis making use of unequal-mass stars yields consistent conclusions. In particular, it yields a boson-mass of $8.57^{+0.64}_{-0.67}\times 10^{-13}$\,eV.

\begin{table}
    \centering
    \begin{center}
        %\begin{tabularx}{\columnwidth}{>{\raggedright\arraybackslash}Xrrr}
        \begin{tabular}{c|cc}
            %\hline 
            %\hline
            \rule{0pt}{3ex}%
            Waveform model  & \multicolumn{2}{c}{$\log\cal{B}$}  \\
            \hline
            \rule{0pt}{3ex}%
                            &  {\footnotesize Comoving Volume} & Uniform \\
                            \rule{0pt}{3ex}%
            BBH (\nrsur)  & 89.6  &  89.7    \\
            \rule{0pt}{3ex}%
            Proca $q=1$ $h(t)$  & 90.6  & 93.2     \\
            \rule{0pt}{3ex}%
            Proca $q\neq1$ $h(t)$  & 91.4  & 94.0   \\
            \rule{0pt}{3ex}%
            Proca $q=1$  $\Psi_4(t)$ & 90.4  & 93.0     \\
            \rule{0pt}{3ex}%
            Proca $q\neq1$  $\Psi_4(t)$ & 91.3  & 93.9   
        \end{tabular}
        \caption{\textbf{Model selection for GW190521} We report the natural Log Bayes factor obtained for our different waveform models. For the Proca star merger model, analyses done under the classical strain formalism and our $\Psi_4$-formalism are equivalent.}
        \label{tab:logb_21g}
    \end{center}
\end{table}

\begin{table*}
    \centering
    \begin{tabular}{lcc|cc}
        \rule{0pt}{3ex}%
        Parameter  & \multicolumn{2}{c}{Waveform model}  \\
        \hline
        \rule{0pt}{3ex}%
                   & \multicolumn{2}{c}{$q=1$} & \multicolumn{2}{c}{$q \neq 1$} \\
                   \hline
                   \rule{0pt}{3ex}%
                   &  $h(t)$ & $\psi_4(t)$ & $h(t)$ & $\psi_4(t)$ \\
                   \rule{0pt}{3ex}%
        Total Mass $[M_\odot]$                  & $258^{+12}_{-10}$  & $260^{+8}_{-9}$ &  $260^{+9}_{-10}$  &  $262^{+8}_{-7}$ \\
        \rule{0pt}{3ex}%$$$$
        Total Source-frame Mass $[M_\odot]$     & $232^{+12}_{-17}$  & $233^{+12}_{-16}$ & $232^{+14}_{-15}$ & $232^{+14}_{-15}$   \\
        \rule{0pt}{3ex}%$$$$
        Primary Source-frame Mass $[M_\odot]$   & $116^{+6}_{-9}$    & $117^{+5}_{-8}$  & $120^{+10}_{-9}$ & $120^{+9}_{-10}$ \\
        \rule{0pt}{3ex}%$$$$
        Secondary Source-frame Mass $[M_\odot]$ & $116^{+6}_{-9}$    & $117^{+5}_{-8}$  & $111^{+8}_{-7}$  & $111^{+8}_{-7}$  \\
        \rule{0pt}{3ex}%
        Luminosity distance [Mpc]               & $541^{+305}_{-176}$   & $544^{+296}_{-163}$  & $592^{+358}_{-262}$ & $618^{+360}_{-244}$ \\
        \rule{0pt}{3ex}%
        Inclination [rad]                       & $0.83^{+0.23}_{-0.45}$ & $0.85^{+0.23}_{-0.36}$ & $0.67^{+0.36}_{-0.45}$ & $0.67^{+0.33}_{-0.45}$ \\
        \rule{0pt}{3ex}%$$
        Final spin                              & $0.69^{+0.04}_{-0.04}$ & $0.70^{+0.04}_{-0.03}$ & $0.70^{+0.03}_{-0.05}$ & $0.71^{+0.02}_{-0.05}$   \\
        \rule{0pt}{3ex}%
        Primary field frequency $\omega_1/\muB$    & $0.890^{+0.018}_{-0.018}$   & $0.895^{+0.015}_{-0.015}$ & $0.895$                   & $0.895$ \\
        \rule{0pt}{3ex}%$$$$$
        Secondary field frequency $\omega_2/\muB$  & $0.890^{+0.018}_{-0.018}$   & $0.895^{+0.015}_{-0.015}$ & $0.900^{+0.018}_{-0.018}$ & $0.905^{+0.015}_{-0.018}$   \\
        \rule{0pt}{3ex}%$$$$$
        Boson mass $\muB\ [10^{-13}\,{\rm eV}]$     & $8.80^{+0.76}_{-0.93}$      & $8.60^{+0.63}_{-0.62}$    & $8.63^{+0.70}_{-0.68}$    & $8.57^{+0.64}_{-0.67}$ \\
        \rule{0pt}{3ex}%
        \rule{0pt}{3ex}%
    \end{tabular}
    \caption{\textbf{Parameters of GW190521 as a head-on Proca star merger} We report median values together with symmetric $90\%$ credible intervals under the scenario of an equal-mass, equal-spin merger and under our exploratory unequal-mass model. Columns labelled by $h(t)$ correspond to a ``classical'' analysis performed with strain-data and templates; while those labelled by \ncor{$\psi_4(t)$} make use of $\Psi_4(t)$-data and templates. We quote results corresponding to a distance prior uniform in co-moving volume.}
    \label{tab:pe}
\end{table*}

\section{Results on real data II: the trigger S200114f}

\begin{figure*}
    \begin{center}
        \includegraphics[width=0.32\textwidth]{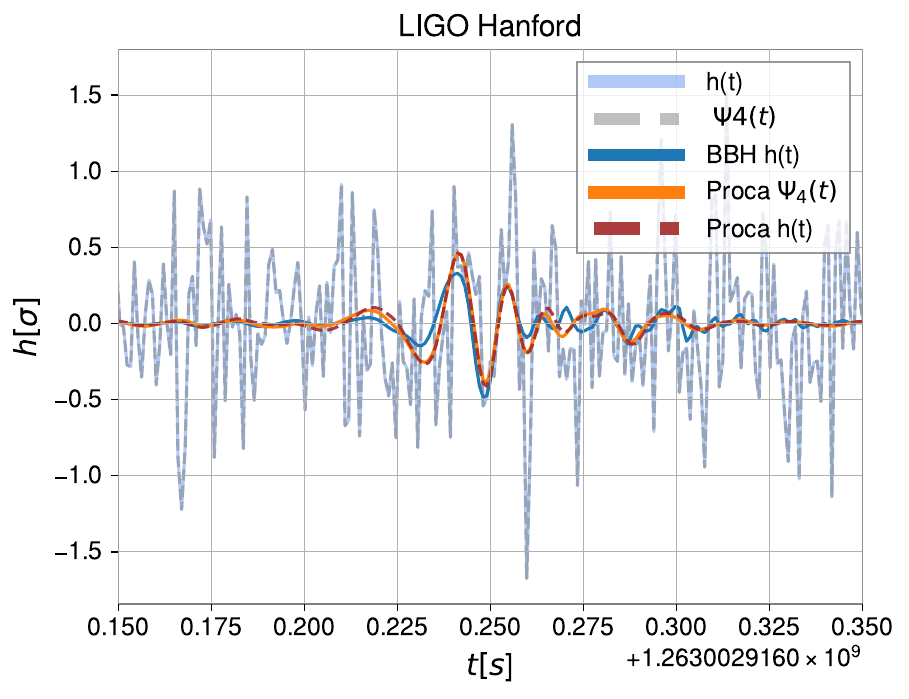}
        \includegraphics[width=0.32\textwidth]{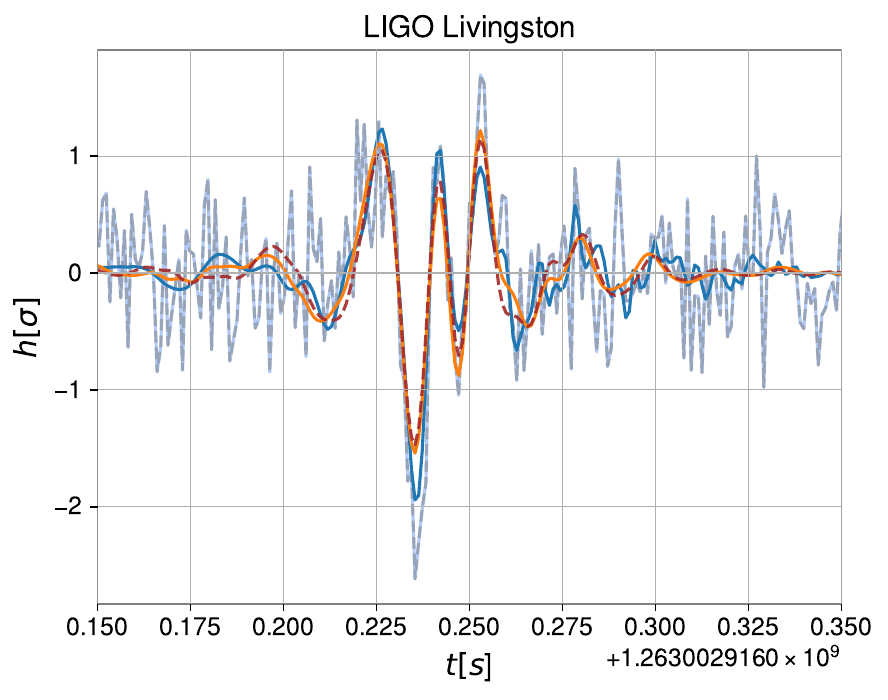}
        \includegraphics[width=0.32\textwidth]{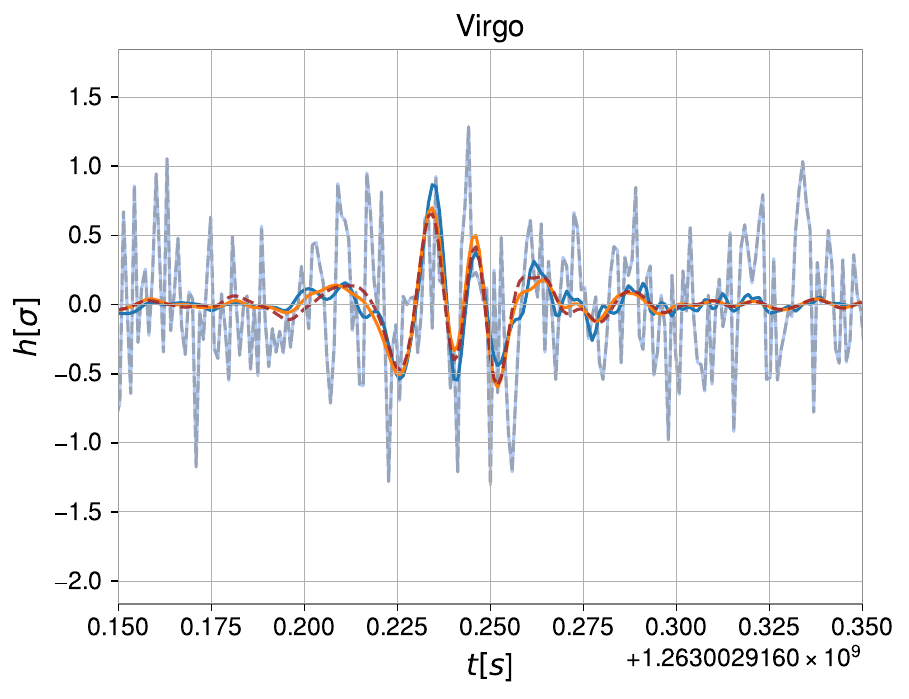}
        \caption{
        \textbf{Whitened detector data and maximum likelihood waveforms for S200114f} We show the whitened strain $h(t)$ (light blue) and $\Psi_4(t)$ (grey) detector data around the time of GW190521 together with the maximum likelihood waveforms returned by the BBH model \nrsur (blue) and by our equal-mass Proca star merger simulations obtained through strain (brown) and $\Psi_4(t)$ analyses (orange). Unlike in Fig. 9 for the case of GW190521, in this case, the brown and orange waveforms show visible differences particularly visible in the early part of the waveforms. These translate into a worse fit to the data in the brown case and qualitatively different conclusions in terms of model selection.}.
        \label{fig:maxL_14f}
    \end{center}
\end{figure*}

Finally, we show a real data example for which our framework makes an important difference. The trigger S200114f is a LIGO-Virgo high-mass trigger detected by a model agnostic \ncor{search} that identifies coherent excess power across the detector network, known as \texttt{coherent WaveBurst} \cite{Klimenko:2015ypf}, with a highly significant false-alarm-rate of 1/17 yr \cite{O3IMBH}. This trigger has, however, challenged existing waveform models. In particular, while the LVK Collaboration analysed S200114f under three BBH waveform models, no pair of these models returned consistent parameter estimates. Far from indicating that this trigger is not of astrophysical origin, this a symptom that the mentioned waveform models disagree in the regions of the parameter space where they best reproduce the signal. \ncor{However}, this trigger is morphologically consistent with a type of noise transients known as Tomte glitches \cite{Merritt2021_tomtes}. In this situation, the LVK decided not to classify it as a confident or catalogued detection but, importantly, nor was it classified as a noise trigger.\\

The above characteristics make S200114f a tantalising candidate to compare to our  simulation catalogue of Proca-star mergers, in the same way we previously treated GW190521. We note, however, that because we find that S200114f was poorly reproduced by the small simulation sets mentioned above, here we use an enhanced bank of nearly 759 simulations spanning a grid in the two-star oscillation frequencies the space $\omega_{1,2}\ncor{/\mu_{\rm B}} \in [0.80,0.93]$. Due to this enhancement, we use priors that slightly differ from those of the GW190521 analysis, which we specify below. Just as for the case of GW190521, we perform our comparison to the BBH model \texttt{NRSur7dq4} within the strain framework while for the Proca-star case we use both the strain and $\Psi_4$ formalisms. Unlike in the case of GW190521, however, we find that both methods return significantly different results that arise from the fact that integration/filtering issues affect the best-fitting strain waveforms.\\

\subsubsection*{Bayesian Priors}

For the PSM model, we impose the same priors as in our GW190521 study, with the exception of the field frequencies. \ncor{For these, we impose a prior uniform across the triangle defined by $\omega_{1,2}\ncor{/\mu_{\rm B}} \in [0.80,0.93]$, with $\omega_{1}\ncor{/\mu_{\rm B}} \geq \omega_{2}\ncor{/\mu_{\rm B}}$. Finally, for the BBH model, we impose the same priors as for the GW190521 case. However, from the four mass-ratio priors we discuss we retain the one yielding the largest Bayesian evidence for the $\Psi_4$ run. The goal of this is to be as conservative as we can in our statements in favour of the existence of Proca stars. These priors are identical to those imposed in \cite{psi4_observations}, where we refer the reader to for further details.}.

\subsection{Model selection}

Fig. \ref{fig:maxL_14f} shows the whitened time series from the three detectors around the time of S200114f. Overlaid, we show the maximum likelihood waveforms for the BBH case and for the Proca-star merger case, the latter both in terms of strain and $\Psi_4$. Unlike for the case of GW190521, the latter two (brown and orange) clearly differ. This difference is particularly visible in the early pre-merger part of the signal, which is most prone to be affected by integration artefacts and choices of the integration frequency cutoff. Additional differences are also observable in the late ringdown part. While visually mild,  such disagreement drives dramatically different values of the corresponding likelihood, which is 6 e-folds larger in the $\Psi_4$ case. This, in turn, has a great impact on model selection, \ncor{as shown in Table \ref{tab:bayes_14f}}. While under the $\Psi_4$ formalism we obtain $\log {\cal{B}} = 2.0$, this is reduced to $\log {\cal{B}} = -7.6$ when using strain templates, in the case where we use our distance prior uniform in co-moving volume. In other words, while the trigger is slightly preferred as a Proca-star merger under the artefact-free $\Psi_4$ analysis, such an option is conclusively discarded under the strain analysis due to the artefacts arisen during the waveform integration process. We note that this result is qualitatively consistent with that returned by the noise-free injection study described in section VI. C. In particular, the $\log {\cal{B}}$ reported in the first ($\Psi_4$ vs. $\Psi_4$) and fifth ($\Psi_4$ vs. $\delta^{2}h_F$) rows of the fourth column of Table \ref{tab:injections} differ by 6.1 units, as compared to the 9.2 units we obtain in real data. Finally, a similarly dramatic effect is observed when we use a prior uniform in distance. In this case, the usage of strain waveforms causes a reduction from $\log {\cal{B}}=5.3$ to  $\log {\cal{B}}=-0.3$, i.e., from a strong preference for the Proca-star scenario to rather equal preference for both scenarios.\\

\begin{table}
    \centering
    \begin{center}
        %\begin{tabularx}{\columnwidth}{>{\raggedright\arraybackslash}Xrrr}
        \begin{tabular}{c|cc}
            %\hline 
            %\hline
            \rule{0pt}{3ex}%
            Waveform model  & \multicolumn{2}{c}{$\log\cal{B}$}  \\
            \hline
            \rule{0pt}{3ex}%
                            &  {\footnotesize Comoving Volume} & Uniform \\
                            \rule{0pt}{3ex}%
            BBH (\nrsur)  & 69.1  &  71.0    \\
            \rule{0pt}{3ex}%
            Proca $h(t)$  & 61.5  & 69.7     \\
            \rule{0pt}{3ex}%
            Proca $\Psi_4(t)$ & 71.1  & 76.3     \\ 
        \end{tabular}
        \caption{\textbf{Model selection for s200114f} We report the natural Log Bayes factor obtained for our different waveform models under different signal models and distance priors.}
        \label{tab:bayes_14f}
    \end{center}
\end{table}

\subsection{Parameter estimation}

Finally, for the sake of completeness, Table \ref{tab:pe_14f} shows our parameter estimates for S200114f under both the strain and $\Psi_4$ analyses. First, clear differences arise in the estimated luminosity distance, total redshifted mass, star frequency and spin parameters. These translate into biases in the boson mass and source-frame mass estimates. In particular, the boson-mass estimate from the strain formalism is highly consistent with that of GW190521 while it becomes highly inconsistent if using the integration-error-free $\Psi_4$ waveforms.

\begin{table}
    \centering
    \begin{tabular}{lcc}
        \rule{0pt}{3ex}%
        Parameter  & \multicolumn{2}{c}{Waveform model}  \\
        \hline
        \rule{0pt}{3ex}%
                   &  $h(t)$ & $\psi_4(t)$   \\
                   \hline

                   \rule{0pt}{3ex}%
        Total Mass $[M_\odot]$                  & $233^{+15}_{-29}$  & $215^{+18}_{-15}$  \\
        \rule{0pt}{3ex}%$$$$
        Total Source-frame Mass $[M_\odot]$     & $228^{+17}_{-29}$  & $207^{+16}_{-14}$   \\
        \rule{0pt}{3ex}%$$$$
        Primary Source-frame Mass $[M_\odot]$   & $123^{+8}_{-14}$    & $119^{+9}_{-14}$   \\
        \rule{0pt}{3ex}%$$$$
        Secondary Source-frame Mass $[M_\odot]$ & $107^{+7}_{-17}$    & $88^{+16}_{-7}$    \\
        \rule{0pt}{3ex}%
        Luminosity distance [Mpc]               & $88^{+109}_{-28}$   & $152^{+73}_{-61}$   \\
        \rule{0pt}{3ex}%
        Inclination [rad]                       & $1.03^{+0.32}_{-0.45}$ & $0.91^{+0.50}_{-0.24}$  \\
        \rule{0pt}{3ex}%$$
        Final spin                              & $0.63^{+0.07}_{-0.01}$ & $0.66^{+0.03}_{-0.04}$   \\
        \rule{0pt}{3ex}%
        Primary field frequency $\omega_1/\muB$    & $0.887^{+0.033}_{-0.019}$   & $0.919^{+0.006}_{-0.043}$  \\
        \rule{0pt}{3ex}%$$$$$
        Secondary field frequency $\omega_2/\muB$  & $0.833^{+0.040}_{-0.025}$   & $0.810^{+0.062}_{-0.010}$    \\
        \rule{0pt}{3ex}%$$$$$
        Boson mass $\muB\ [10^{-13}\,{\rm eV}]$     & $9.67^{+0.67}_{-0.49}$      & $10.20^{+0.68}_{-0.55}$  \\
        \rule{0pt}{3ex}%$$$$$
        $\log{\cal{B}}^{\text{Proca-star}}_{\text{BBH}} $     & -7.6      & 2.0  \\
        
        \rule{0pt}{3ex}%
        \rule{0pt}{3ex}%
    \end{tabular}
    \caption{\textbf{Parameters of S200114f as a head-on Proca-star merger} We report median values together with symmetric $90\%$ credible intervals. The column labelled by $h(t)$ corresponds to a ``classical'' analysis performed with strain-data and templates; while that labelled by $\psi_4(t)$ makes use of $\Psi_4(t)$-data and templates. We quote results obtained under a distance prior uniform in co-moving volume.}
    \label{tab:pe_14f}
\end{table}

\section{Discussion of results: integrated strain vs. Newman-Penrose scalar}
Given our results, both on synthetic signals and on real data, the question arises of: what is special about the waveforms reproducing S200114f that makes strain waveforms problematic, as opposed to the case of GW190521?

The answer is: nothing in principle. As stated throughout the paper, obtaining ``integration-error free'' strain waveforms depends on a series of human choices (in particular that of $\omega_0$) that are only reasonably well guided for the case of quasi-circular mergers. In other cases, obtaining clean strain waveforms through FFI
(\textit{if possible at all}) involves a way more convoluted trial-and-error process. Moreover, we understand that in the absence of a ``true'' reference waveform, concluding that the obtained waveform is correct can only be done through the comparison to the Newman-Penrose one, similarly to what we do in the right panels of Figs. \ref{fig:whitening_21gtemplate-1} and \ref{fig:whitening_14ftemplate}.

In this situation, we can only make the \textit{everything but scientific argument} that the integration choices that happened to work correctly for the waveforms best fitting GW190521 (making our results in \cite{Bustillo:2021proca1} safe from integration artefacts), did not work for the waveforms best fitting S200114f. \ncor{This is, at least partially, explained by the fact that our choice of M$\omega_0$ does remove more true signal power for S200114f than for the GW190521-like injection.} In this situation, it could be argued that better strain waveforms may have been obtained if further exploration of $\omega_0$-choices was performed, \ncor{although there is no guarantee that this process would lead to a driftless waveform containing all the true signal power. In fact, FFI does by definition eliminate some true signal power even in quasi-circular cases}. In our view, this exactly exemplifies the huge advantage that the usage of the Newman-Penrose scalar has over that of the integrated strain: the resulting waveforms are ``uniquely defined'', with no choices to be made beyond those pertaining the specific configuration of the numerical simulation itself.

\section{Conclusions}
Extracting the properties of GW sources requires accurate waveform templates that can be compared to detector data. NR provides the most precise way to obtain such templates and it is often the only way. Computing GW strain waveforms from $\psi_4$ outputted by NR simulations that can be compared to the strain detector data is a non-trivial process subject to well-known systematics that can impact the physical interpretation of the source. Moreover, easing these errors is a rather artisan process subject to human choices that are not always obvious or even well-motivated depending on the considered type of source. This is particularly problematic for some of the most astrophysically interesting sources LIGO and Virgo are starting to observe, like precessing mergers \cite{GWTC3,Hannam_nature_precession}, or may observe in the future observing runs, like eccentric mergers or dynamical captures, for which there is no monotonic relation between GW frequency and time. We note that, even in cases where such relation exists, typical systematic errors of $\sim 1\%$ in amplitude will always exist. Moreover, these are in practice impossible to know because the true waveform is not known \cite{Pollney_Reissweig}. By taking second-order finite differences on the detector strain data, we have presented a data analysis framework that allows to directly compare GW data to $\psi_4$, removing the need to extract the GW strain from numerical simulations and the associated systematic errors. We have shown that our framework is equivalent to the traditional strain one modulo the potential systematic errors present in the strain waveforms. Therefore, given that $\Psi_4$-waveforms have one less layer of systematic errors than strain ones, classical strain analyses will, at best, be as faithful as $\Psi_4$-based ones.\\ 

As a demonstration of our framework in real data analysis, we have first repeated our previous study comparing GW190521 to numerically simulated strain waveforms from Proca star mergers presented in Ref.~\cite{Bustillo:2021proca1}, but using the direct $\Psi_4$ outputted by our numerical simulations. We obtain results completely consistent with the original ones, which is indicative that our strain waveforms best-fitting GW190521 suffered, at most, from mild integration errors that did not impact our original analysis. Second, we have analysed the high-mass trigger S200114f using an enhanced catalog of Proca-star merger simulations reported in \cite{psi4_observations}. In this case, we find that the usage of strain waveforms -- affected by integration errors -- has a huge impact in the interpretation of this signal, yielding conclusions that differ dramatically with respect to those obtained by using error-free $\Psi_4$ waveforms. %In a forthcoming paper, we use this new formalism to compare high-mass events from LIGO-Virgo-KAGRA to an extended catalogue of 759 simulations of head-on mergers of Proca stars.
\\

Our framework removes the need to obtain strain waveforms from numerical relativity simulations, removing a complete layer of systematic errors. We note that while we have focused on the case of short numerical relativity simulations with rather ``exotic'' dynamics, the integration errors worsen with increasing waveform length. In particular, this makes such errors particularly troublesome in the task of constructing hybrid numerical-relativity - post-Newtonian waveforms \cite{OhmeThesis} that are matched are early times \cite{Pollney_Reissweig}. 

While we have discussed our procedure under the prevalent scenario where the transverse-traceless (TT) gauge is considered, its application to alternative gauge-independent formulations \cite{KoopFinn} is, in principle, straight forward. Similarly, while we have demonstrated our framework in the context of parameter inference and model selection, this is trivially applicable to the case of actual matched-filter searches for GW signals  ~\cite{Usman:2015kfa,Messick:2016aqy,Harry:2017weg,MBTA,SPIIR,Chandra2022}. Finally, we note that since LVK results have so far been obtained under the assumption of quasi-circular mergers, we have no reasons to believe such results may be affected by the errors we have discussed here.

\section*{Acknowledgements} 

\textcolor{black}{We plan to publish and mantain our code to perform gravitational-wave data analysis using the Newman-Penrose scalar, within the software \texttt{Bilby} \cite{Ashton:2018jfp,pbilby} at \cite{BilbyProca}.}
The analysed LIGO-Virgo data and the corresponding power spectral densities, in their strain versions, are publicly available at the online Gravitational-Wave Open Science Center \cite{SoftwareX,open_data}.This research has made use of data or software obtained from the Gravitational Wave Open Science Center (gwosc.org), a service of LIGO Laboratory, the LIGO Scientific Collaboration, the Virgo Collaboration, and KAGRA. LIGO Laboratory and Advanced LIGO are funded by the United States National Science Foundation (NSF) as well as the Science and Technology Facilities Council (STFC) of the United Kingdom, the Max-Planck-Society (MPS), and the State of Niedersachsen/Germany for support of the construction of Advanced LIGO and construction and operation of the GEO600 detector. Additional support for Advanced LIGO was provided by the Australian Research Council. Virgo is funded, through the European Gravitational Observatory (EGO), by the French Centre National de Recherche Scientifique (CNRS), the Italian Istituto Nazionale di Fisica Nucleare (INFN) and the Dutch Nikhef, with contributions by institutions from Belgium, Germany, Greece, Hungary, Ireland, Japan, Monaco, Poland, Portugal, Spain. KAGRA is supported by Ministry of Education, Culture, Sports, Science and Technology (MEXT), Japan Society for the Promotion of Science (JSPS) in Japan; National Research Foundation (NRF) and Ministry of Science and ICT (MSIT) in Korea; Academia Sinica (AS) and National Science and Technology Council (NSTC) in Taiwan. JCB is supported by a fellowship from ``la Caixa'' Foundation (ID
100010434) and from the European Union’s Horizon 2020 research and innovation programme under the Marie Skłodowska-Curie grant agreement No 847648. The fellowship code is LCF/BQ/PI20/11760016. JCB is also supported by the research grant PID2020-118635GB-I00 from the Spain-Ministerio de Ciencia e Innovaci\'{o}n. KC acknowledges the MHRD, Government of India, for the fellowship support.
JAF is supported by the Spanish Agencia Estatal de Investigaci\'on  (grants PGC2018-095984-B-I00 and PID2021-125485NB-C21) and by the  Generalitat  Valenciana  (PROMETEO/2019/071). This work is supported by the Center for Research and Development in Mathematics and Applications (CIDMA) 
through the Portuguese Foundation for Science and Technology (FCT - Funda\c {c}\~ao para a Ci\^encia e a Tecnologia), reference UIDB/04106/2020, and by national funds (OE), through FCT, I.P., in the scope of the framework contract foreseen in the numbers 4, 5 and 6 of the article 23, of the Decree-Law 57/2016, of August 29, changed by Law 57/2017, of July 19.
We also acknowledge support  from  the  projects  PTDC/FIS-OUT/28407/2017,  CERN/FIS-PAR/0027/2019, PTDC/FIS-AST/3041/2020,  CERN/FIS-PAR/0024/2021 and 2022.04560.PTDC. NSG is supported by the Spanish Ministerio de Universidades, through a María Zambrano grant (ZA21-031) with reference UP2021-044, funded within the European Union-Next Generation EU.   
This work has further been supported by the European Union’s Horizon 2020 research and innovation (RISE) programme H2020-MSCA-RISE-2017 Grant No. Fu\text{NF}iCO-777740 and by the European
Horizon Europe staff exchange (SE) programme HORIZON-
MSCA-2021-SE-01 Grant No. NewFunFiCO-101086251. 
We acknowledge the use of IUCAA LDG cluster Sarathi for the computational/numerical work. The authors acknowledge computational resources provided by the CIT cluster of the LIGO Laboratory and supported by National Science Foundation Grants PHY-0757058 and PHY0823459; and the support of the NSF CIT cluster for the provision of computational resources for our parameter inference runs. This material is based upon work supported by NSF's LIGO Laboratory which is a major facility fully funded by the National Science Foundation. This manuscript has LIGO DCC number P2200114. 

%\clearpage
\appendix

\section{Relation between the Fourier transform of the second order difference and the second derivative of a time series} \label{app:relation}

\begin{theorem}
    Given a continuous-time time series $x(t)$ where $t \in (0, T)$ of duration $T$ and the sampled time series $x[m]$ of sampling interval $\Delta t$ where $m = 0, 1, ..., M - 1$, i.e.\ $M\Delta t = T$, if the Fourier transform of the continuous-time time series and the discrete Fourier transform of the sampled time series are equivalent i.e.\ $\tilde{x}(k\Delta f) = \tilde{x}[k]$ where
\begin{equation}
    \tilde{x}(f) = \int_{0}^{T}x(t)e^{-i2\pi ft} dt \,,
\end{equation}
\begin{equation}
    \tilde{x}[k] = \Delta t \sum_{\ncor{m}=0}^{M - 1}x[m]e^{-i2\pi mk / M} \,,
\end{equation}
and $\Delta f = 1 / T$, and the second derivative of $x(t)$ exists at every point in $(0, T)$, then the Fourier transform of the second derivative of $x(t)$ and the discrete Fourier transform of the second difference of $x[m]$ are related by
    \begin{equation}
        \label{eq:app_relation}
        \widetilde{\delta^{2}x}[k] = 
        \frac{1-\cos(2\pi k \Delta f \Delta t)}{2\pi^{2}(k\Delta f \Delta t)^{2}}\widetilde{x''}(k \Delta f)
    \end{equation}
    where $x''(t)$ is the second derivative of $x(t)$, and $\delta^{2}x[m]$ is the second difference of $x[m]$ defined by
    \begin{equation}
        \delta^{2}x[m] = 
        \frac{x[m + 1] - 2x[m] + x[m - 1]}{(\Delta t)^{2}}\,.
    \end{equation}
\end{theorem}

\begin{proof}
    The Fourier transform of $x(t)$ and $x''(t)$ are related by
    \begin{equation}
        \label{eq:app_xf}
        \widetilde{x''}(f) = -4\pi^{2}f^{2}\widetilde{x}(f)\,.
    \end{equation}
    The discrete Fourier transform of $x[m]$ and $(\delta^{2}x)[m]$ are related by
    \begin{equation}
        \label{eq:app_xk}
        \widetilde{(\delta^{2}x)}[k] = \frac{2(\cos(2\pi k /M) - 1)}{(\Delta t)^{2}}\,\widetilde{x}[k] \,.
    \end{equation}
    If $\tilde{x}(k\Delta f) = \tilde{x}[k]$, from Eq.~\eqref{eq:app_xf} and Eq.~\eqref{eq:app_xk}, we have
     \begin{equation}
        \widetilde{\delta^{2}x}[k] = 
        \frac{1-\cos(2\pi k \Delta f \Delta t)}{2\pi^{2}(k\Delta f \Delta t)^{2}}\widetilde{x''}(k \Delta f)\,.
    \end{equation}
\end{proof}
\textcolor{black}{
When assuming $\tilde{x}(k\Delta f) = \tilde{x}[k]$, it is important to note that aliasing and spectral leakage are intrinsic to discrete Fourier transforms. These issues also arise in ``regular'' GW data analysis when computing the discrete Fourier transform of strain data. Proper data windowing, such as employing a Tukey window as we do, mitigates these errors by \textcolor{black}{tapering} data to zero at the ends of the segment while preserving the GW signal segment. Importantly, our method does not introduce any new sources of systematic errors. Therefore, the errors discussed above are equivalent to those encountered in regular GW data analysis and are considered for the sake of rigor.}

\textcolor{black}{In this context, the discrete Fourier transform of the sampled time series closely approximates the Fourier transform of the continuous-time series, validating Eq.~\eqref{eq:app_relation}. Nevertheless, accuracy depends on precise discrete Fourier transform usage, necessitating caution regarding aliasing and spectral leakage. Appropriate data windowing is essential when applying the discrete Fourier transform.
}

\section{Distribution of second-differenced noise}
\label{app:distr_delta_2}

\begin{theorem}
    \label{thm:distr_delta2}
    Let $\boldsymbol{x}\in\mathbb{R}^{M}$ be a discrete-time Gaussian process such that the mean is $\mathbb{E}\left[\boldsymbol{x}\right] = \boldsymbol{\mu}$ and the covariance is $\mathbb{E}\left[(\boldsymbol{x} - \boldsymbol{\mu})(\boldsymbol{x} - \boldsymbol{\mu})^{T}\right] = \boldsymbol{\Sigma}$. Let $\delta^{2}\boldsymbol{x}\in\mathbb{R}^{M}$ be the second difference of $\boldsymbol{x}$ defined by
    \begin{equation}
        (\delta^{2}x)[m] = \frac{x[m+1] - 2x[m] + x[m-1]}{(\Delta t)^{2}}
    \end{equation}
    where $x[m]$ is the $m$-th element of $\vb*x$, and $M > 2$ and $\Delta t$ are respectively the total length and the sampling interval of the discrete-time process respectively. A periodic boundary condition is imposed such that $x[0] = x[M]$ and $x[M+1] = x[1]$, then $\delta^{2}\boldsymbol{x}$ is a discrete-time Gaussian process with mean $\boldsymbol{\mu}_{2} = \boldsymbol{T}\boldsymbol{\mu}$ and covariance $\boldsymbol{\Sigma}_{2} = \boldsymbol{T}\boldsymbol{\Sigma}\boldsymbol{T}$ where $\boldsymbol{T}$ is a $M\times M$ matrix with entries $T_{j,j} = -2 / (\Delta t)^{2}$, $T_{j,j+1} = T_{j,j-1} = 1 / (\Delta t)^{2}$ with a periodic boundary condition imposed on the matrix index i.e.\ index $0$ implies index $M$ and index $M+1$ implies index $1$, and otherwise zero.
\end{theorem}

\begin{proof}
    The probability density function of the discrete-time stationary Gaussian process is
    \begin{equation}
        p(\boldsymbol{x}) = \frac{1}{(2\pi)^{M/2}|\boldsymbol{\Sigma}|^{1/2}}\exp\left(-\frac{1}{2}\left(\boldsymbol{x}-\boldsymbol{\mu}\right)^{T}\boldsymbol{\Sigma}^{-1}\left(\boldsymbol{x} - \boldsymbol{\mu}\right)\right)\,.
    \end{equation}
    The second difference of $\boldsymbol{x}$ can be regarded as a linear transformation of $\boldsymbol{x}$. Since 
    
    \begin{align}
        \begin{bmatrix}
            (\delta^{2}x)[1] \\
            (\delta^{2}x)[2] \\
            \vdots \\
            (\delta^{2}x)[M]
        \end{bmatrix}
        &= \begin{bmatrix}
                \frac{1}{(\Delta t)^{2}}\left(x[2] - 2x[1] + x[M]\right) \\
                \frac{1}{(\Delta t)^{2}}\left(x[3] - 2x[2] + x[1]\right) \\
                \vdots \\
                \frac{1}{(\Delta t)^{2}}\left(x[1] - 2x[M] + x[M-1]\right)
            \end{bmatrix} \nonumber\\
        &= \frac{1}{(\Delta t)^{2}}
        \begin{bmatrix}
            -2 & 1 & 0 & \cdots & 0 & 1 \\
            1 & -2 & 1 & \cdots & 0 & 0 \\
            \vdots & \vdots & \vdots & \ddots & \vdots & \vdots \\
            1 & 0 & 0 & \cdots & 1 & -2
        \end{bmatrix}
        \begin{bmatrix}
            x[1] \\
            x[2] \\
            \vdots \\
            x[M]
        \end{bmatrix}\,,
    \end{align}
    
    we can write $\delta^{2}\boldsymbol{x} = \boldsymbol{T}\boldsymbol{x}$ where
    \begin{equation}
        \boldsymbol{T} = 
        \frac{1}{(\Delta t)^{2}}
        \begin{bmatrix}
            -2 & 1 & 0 & \cdots & 0 & 1 \\
            1 & -2 & 1 & \cdots & 0 & 0 \\
            \vdots & \vdots & \vdots & \ddots & \vdots & \vdots \\
            1 & 0 & 0 & \cdots & 1 & -2
        \end{bmatrix} \,. 
    \end{equation}
    
    or $T_{j,j} = -2 / (\Delta t)^{2}$, $T_{j,j+1} = T_{j,j-1} = 1 / (\Delta t)^{2}$ and a periodic boundary condition is imposed on the matrix index i.e.\ index $0$ implies index $N$ and index $N+1$ implies index $1$, and otherwise zero. The probability density function of $\delta^{2}\boldsymbol{x}$ is then \\
    
    \begin{widetext}
        \begin{equation}
            \begin{split}
                q(\delta^{2}\boldsymbol{x}) &= p(\boldsymbol{T}^{-1}\delta^{2}\boldsymbol{x})\left|\frac{\partial (\boldsymbol{T}^{-1}\delta^{2}\boldsymbol{x})}{\partial (\delta^{2}\boldsymbol{x})}\right| \\
                                            &=p(\boldsymbol{T}^{-1}\delta^{2}\boldsymbol{x})\left|\boldsymbol{T}^{-1}\right| \\
                                            &=\frac{1}{\left|\boldsymbol{T}\right|} \frac{1}{(2\pi)^{M/2}|\boldsymbol{\Sigma}|^{1/2}}\exp\left(-\frac{1}{2}\left(\boldsymbol{T}^{-1}\delta^{2}\boldsymbol{x}-\boldsymbol{\mu}\right)^{T}\boldsymbol{\Sigma}^{-1}\left(\boldsymbol{T}^{-1}\delta^{2}\boldsymbol{x} - \boldsymbol{\mu}\right)\right) \\
                                            &= \frac{1}{(2\pi)^{M/2}|\boldsymbol{T}\boldsymbol{\Sigma}\boldsymbol{T}|^{1/2}}\exp\left(-\frac{1}{2}\left(\boldsymbol{T}^{-1}\delta^{2}\boldsymbol{x}-\boldsymbol{\mu}\right)^{T}\boldsymbol{\Sigma}^{-1}\left(\boldsymbol{T}^{-1}\delta^{2}\boldsymbol{x} - \boldsymbol{\mu}\right)\right) \\
                                            &= \frac{1}{(2\pi)^{M/2}|\boldsymbol{T}\boldsymbol{\Sigma}\boldsymbol{T}|^{1/2}}\exp\left(-\frac{1}{2}\left(\delta^{2}\boldsymbol{x}-\boldsymbol{T}\boldsymbol{\mu}\right)^{T}(\boldsymbol{T}^{-1})^{T}\boldsymbol{\Sigma}^{-1}\boldsymbol{T}^{-1}\left(\delta^{2}\boldsymbol{x} - \boldsymbol{T}\boldsymbol{\mu}\right)\right) \\
                                            &= \frac{1}{(2\pi)^{M/2}|\boldsymbol{T}\boldsymbol{\Sigma}\boldsymbol{T}|^{1/2}}\exp\left(-\frac{1}{2}\left(\delta^{2}\boldsymbol{x}-\boldsymbol{T}\boldsymbol{\mu}\right)^{T}(\boldsymbol{T}\boldsymbol{\Sigma}\boldsymbol{T}^{T})^{-1}\left(\delta^{2}\boldsymbol{x} - \boldsymbol{T}\boldsymbol{\mu}\right)\right) \\
                                            &= \frac{1}{(2\pi)^{M/2}|\boldsymbol{T}\boldsymbol{\Sigma}\boldsymbol{T}|^{1/2}}\exp\left(-\frac{1}{2}\left(\delta^{2}\boldsymbol{x}-\boldsymbol{T}\boldsymbol{\mu}\right)^{T}(\boldsymbol{T}\boldsymbol{\Sigma}\boldsymbol{T})^{-1}\left(\delta^{2}\boldsymbol{x} - \boldsymbol{T}\boldsymbol{\mu}\right)\right)
            \end{split}
        \end{equation}
    \end{widetext}
    where we have used $\boldsymbol{T}^{T} = \boldsymbol{T}$ since $\boldsymbol{T}$ is a symmetric matrix. Therefore, $\delta^{2}\boldsymbol{x}$ is a discrete-time Gaussian process with mean $\boldsymbol{T}\boldsymbol{\mu}$ and covariance $\boldsymbol{T}\boldsymbol{\Sigma}\boldsymbol{T}$.
\end{proof}

\begin{lemma}
    Let $\boldsymbol{x}\in\mathbb{R}^{M}$ be a stationary discrete-time Gaussian process such that the mean of each discrete point is $\mathbb{E}\left[x[m]\right] = \mu$ and the autocovariance is $K_{XX}[\tau] = \mathbb{E}\left[(x[m] - \mu)( x[m+\tau] - \mu)\right]$, then $\delta^{2}\boldsymbol{x}$ is also a stationary discrete-time Gaussian process with mean $\boldsymbol{\mu}_{2} = \boldsymbol{0}$ and autocovariance $\boldsymbol{K_{XX,2}}$:
    \begin{widetext}
        \begin{equation}
            K_{XX,2}[\tau]\\
            = \frac{1}{(\Delta t)^{4}}
            \left(
                6K_{XX}[\tau] - 
                4K_{XX}[\tau - 1] +
                K_{XX}[\tau - 2] -
                4K_{XX}[\tau + 1] +
                K_{XX}[\tau + 2]
            \right)\,.
        \end{equation}
    \end{widetext}
    The power spectral density $S_{2}[k]$ of $\delta^{2}\boldsymbol{x}$ is related to the power spectral density $S[k]$ of $\boldsymbol{x}$ by
    \begin{equation}
        S_{2}[k] = \frac{1}{(\Delta t)^{4}}\qty(6 - 8\cos(\frac{2\pi k}{M}) + 2\cos(\frac{4\pi k}{M}))S[k]\,.
    \end{equation}
\end{lemma}

\begin{proof}
    By Theorem \ref{thm:distr_delta2}, $\delta^{2}\boldsymbol{x}$ is a discrete-time Gaussian process with mean $\boldsymbol{\mu}_{2} = \boldsymbol{T}\boldsymbol{\mu}$ and covariance $\boldsymbol{\Sigma}_{2} = \boldsymbol{T}\boldsymbol{\Sigma}\boldsymbol{T}$. The mean of $\delta^{2}\boldsymbol{x}$ is then
    \begin{equation}
        \boldsymbol{\mu}_{2} =
        \boldsymbol{T}\boldsymbol{\mu}
        =
        \frac{1}{(\Delta t)^{2}}
        \begin{bmatrix}
            \mu - 2\mu + \mu \\
            \mu - 2\mu + \mu \\
            \vdots \\
            \mu - 2\mu + \mu \\
        \end{bmatrix} \\
        =\boldsymbol{0}\,.
    \end{equation}
    The covariance matrix is
    \begin{widetext}
        \begin{equation}
            \begin{split}
            &\Sigma_{2(n,m)} \\
            &= \sum_{k,l=1}^{N}T_{n,k}\Sigma_{k,l}T_{l,m} \\
            & = \frac{1}{(\Delta t)^{2}}\sum_{l=1}^{N}(\Sigma_{n-1,l} - 2\Sigma_{n,l} + \Sigma_{n+1,l})T_{l,m} \\
            &=\frac{1}{(\Delta t)^{4}}\left[(\Sigma_{n-1,m-1} - 2\Sigma_{n-1,m} + \Sigma_{n-1,m+1})
                -2(\Sigma_{n,m - 1} - 2\Sigma_{n,m} + \Sigma_{n,m + 1}) + 
                (\Sigma_{n+1,m-1} - 2\Sigma_{n+1,m} + \Sigma_{n+1,m+1})
            \right] \\
            &=\frac{1}{(\Delta t)^{4}}
            \left(
                6K_{XX}[|n - m|] - 
                4K_{XX}[|n - m - 1|] +
                K_{XX}[|n - m - 2|] -
                4K_{XX}[|n - m + 1|] +
                K_{XX}[|n - m + 2|]
            \right)
            \end{split}
        \end{equation}
    \end{widetext}
    which depends on the time difference only. $\delta^{2}\boldsymbol{x}$ is therefore also a stationary discrete-time Gaussian process. The autocovariance of $\delta^{2}\boldsymbol{x}$ is therefore
    \begin{widetext}
        \begin{equation}
            K_{XX,2}(\tau)\\
            =
            \frac{1}{(\Delta t)^{4}}
            \left(
                6K_{XX}[\tau] - 
                4K_{XX}[\tau - 1] +
                K_{XX}[\tau - 2] -
                4K_{XX}[\tau + 1] +
                K_{XX}[\tau + 2]
            \right)\,.
        \end{equation}
    \end{widetext}
    Since the power spectral density $S[k]$ is related to the autocorrelation function $R_{XX}[m]$ by
    \begin{equation}
        S[k] = \sum_{m=1}^{M}R_{XX}[m]e^{-i2\pi mk/M}\,,
    \end{equation}
    the power spectral density of the stochastic process $\delta^{2}\boldsymbol{x}$ is therefore
    \begin{widetext}
        \begin{equation}
            \begin{split}
            &S_{2}[k] \\
            &= \sum_{m=1}^{M}R_{XX,2}[m]e^{-i2\pi mk/M}\\
            &= \sum_{m=1}^{M}(K_{XX,2}[m] + \mu_{2}\mu_{2})e^{-i2\pi mk/M}\\
            &= \sum_{m=1}^{M}K_{XX,2}[m] e^{-i2\pi mk/M}\\
            &=
            \frac{1}{(\Delta t)^{4}}
            \sum_{m=1}^{M}\left(
                6K_{XX}[m] - 
                4K_{XX}[m - 1] +
                K_{XX}[m - 2] -
                4K_{XX}[m + 1] +
                K_{XX}[m + 2]
            \right)
            e^{-i2\pi mk/M} \\
            &=\frac{1}{(\Delta t)^{4}}
            (6S[k] - 4S[k]e^{-i2\pi k/M} + S[k]e^{-i4\pi k/M} - 4S[k]e^{i2\pi k/M} + S[k]e^{i4\pi k/M}) \\
            &=\frac{1}{(\Delta t)^{4}}(6 - 8\cos(2\pi k/M) + 2\cos(4\pi k/M))S[k]
            \end{split}
        \end{equation}
    \end{widetext}
    where $S[k]$ is the power spectral density of the stochastic process $\boldsymbol{x}$.
\end{proof}

\section{Appendix: Parameters of simulated signals}

Table \ref{tab:maxl_proca} shows the parameters of the two injections discussed in Section VI. 

\begin{table}[t!]
\centering
\begin{center}
%\begin{ruledtabular}
%\begin{tabularx}{\columnwidth}{>{\raggedright\arraybackslash}Xrr}
\renewcommand{\arraystretch}{1.5}
\begin{tabular}{l|cc|cc}
%\hline
%\hline \\ 
Parameter  & GW190521 & S200114f \\ \hline

Total red-shifted mass $[M_\odot]$  & $269.69$ & $234.31$ 
\\
Inclination $\theta_{JN}$ [rad] & $2.34$ & $1.04$ \\
Azimuth $\varphi$ & $5.11$ & $3.26$ 
\\
Luminosity distance [Mpc] & $246.10$  & $65.71$ & 
\\
Polarization $\psi$ & $1.23$ & $1.51$ 
\\
Right ascension $\alpha$ & $3.94$ & $1.95$ 
\\
Declination $\delta$ & $0.93$ & $0.08$ 
\\
Primary field frequency $\omega_1/\mu_{\rm B}$  & $0.9000$ &   $0.8800$
\\
Secondary field frequency $\omega_2/\mu_{\rm B}$  & $0.8550$  & $0.8325$

\end{tabular}
\end{center}
\caption{\textbf{Parameters of the synthetic signals analysed in section VI.C}. We quote the inclination in terms of the angle between the line-of-sight and the total angular momentum $\theta_{JN}$ as well as the azimuthal angle of the observer $\varphi$ (see Appendix I in \cite{Bustillo:2021proca1}).}
\label{tab:maxl_proca}
\end{table}

%\section*{References}
\bibliography{psi4_formalism.bib}

\end{document}